\documentclass{elsarticle}

\usepackage{hyperref}

\journal{Control Engineering Practice}

\usepackage{amsmath}
\usepackage[english]{babel}

\usepackage[latin1]{inputenc}
\let\textquotedbl="

\usepackage{amsthm}
\newtheorem{thm}{Theorem}
\newtheorem{rem}{Remark}
\newtheorem{lem}{Lemma}
\newtheorem{defn}{Definition}
\newtheorem{cor}{Corollary}

\begin{document}

\begin{frontmatter}

\title{Discrete-time Flatness-based Control\\ of a Gantry Crane\tnoteref{mytitlenote}}
\tnotetext[mytitlenote]{This work has been supported by the Austrian Science Fund (FWF) under grant number P 32151.}

\author{Johannes Diwold\fnref{mymainaddress}}
\ead{johannes.diwold@jku.at}
\address[mymainaddress]{Institute of Automatic Control and Control Systems Technology,	Johannes Kepler University Linz, Altenbergerstraße 66, 4040 Linz, Austria}
\author{Bernd Kolar\fnref{berndaddress}}
\ead{bernd\_kolar@ifac-mail.org}
\address[berndaddress]{Magna Powertrain Engineering Center Steyr GmbH \& Co KG, Steyrer	Str. 32, 4300 St. Valentin, Austria}
\author{Markus Schöberl\fnref{mymainaddress}}
\ead{markus.schoeberl@jku.at}

\begin{abstract}
This article addresses the design of a discrete-time flatness-based tracking control for a gantry crane and demonstrates the practical applicability of the approach by measurement results. The required sampled-data model is derived by an Euler-discretization with a prior state transformation in such a way that the flatness of the continuous-time system is preserved. Like in the continuous-time case, the flatness-based controller design is performed in two steps. First, the sampled-data system is exactly linearized by a discrete-time quasi-static state feedback. Subsequently, a further feedback enforces a stable linear tracking error dynamics. To underline its practical relevance, the performance of the novel discrete-time tracking control is compared to the classical continuous-time approach by measurement results from a laboratory setup. In particular, it turns out that the discrete-time controller is significantly more robust with respect to large sampling times. Moreover, it is shown how the discrete-time approach facilitates the design of optimal reference trajectories, and further measurement results are presented.
\end{abstract}


\end{frontmatter}

\thispagestyle{empty}


\section{Introduction}

The concept of flatness has been introduced by Fliess, L\'{e}vine,
Martin and Rouchon for nonlinear continuous-time systems in the 1990s
(see e.g. \cite{FliessLevineMartinRouchon:1992}, \cite{FliessLevineMartinRouchon:1995}
and \cite{FliessLevineMartinRouchon:1999}). The popularity of flatness
is due to the fact that it allows an elegant solution to motion planning
problems as well as a systematic design of tracking controllers. For
the practical implementation of a flatness-based tracking control
on a digital computer/processor, it is of course important to evaluate
the continuous-time control law with a sufficiently high sampling
rate. If, however, the processing unit or actuators/sensors are limited
to lower sampling rates, a discrete evaluation of a continuous-time
control law may lead to unsatisfactory results. As known from linear
systems theory, an appropriate alternative is to design the controller
directly for a suitable discretization of the system. Moreover, also
other nonlinear control concepts have already been transferred successfully
to the discrete-time framework, see e.g. \cite{KotyckaThoma2021}
for a recent contribution addressing energy-based methods. Motivated
by these considerations, in the present paper we utilize the notion
of discrete-time flatness (\cite{KaldmaeKotta:2013}, \cite{KolarKaldmaeSchoberlKottaSchlacher:2016},
\cite{GuillotMillerioux2020}, \cite{DiwoldKolarSchoeberl2022}) and
prove its practical applicability by designing discrete-time tracking
controllers for the sampled-data model of a gantry crane.

So far, discrete-time flatness-based controllers are only rarely used
in practical applications (except for linear systems, see e.g. \cite{Sira-RamirezAgrawal:2004}).
A main obstacle for a broader usage of the discrete-time approach
is the fact that the flatness of a continuous-time system is not necessarily
preserved by an exact or approximate discretization. Thus, we present
a simple strategy which combines a suitable state transformation and
a subsequent Euler-discretization in such a way that the flatness
of the system is preserved. Although only constructive, the method
is applicable for many well-known nonlinear systems like the gantry
crane, the VTOL aircraft or the induction motor, to mention just a
few. Once a flat sampled-data model has been derived, the design of
tracking controllers can be performed in a systematic way since also
discrete-time flat systems can be exactly linearized by a dynamic
state feedback, see e.g. \cite{DiwoldKolarSchoeberl2022}. For the
considered sampled-data model of the gantry crane, we show that even
an exact linearization by a quasi-static state feedback -- as it
is proposed e.g. in \cite{DelaleauRudolph:1998} for continuous-time
systems -- is possible. With the exactly linearized system, the design
of a tracking control for the stabilization of reference trajectories
becomes a straightforward task. An implementation of both the dynamic
and the quasi-static controller on a laboratory setup has shown that
especially the novel discrete-time quasi-static controller is quite
robust with respect to low sampling rates. Another advantage of the
discrete-time approach becomes apparent when reference trajectories
shall be determined by optimization. In the continuous-time case,
this leads to infinite-dimensional optimization problems, which are
typically reduced to finite-dimensional ones by a discretization of
the space of reference trajectories (using e.g. spline functions).
In the discrete-time case, in contrast, the resulting optimization
problems are inherently finite-dimensional (as long as finite time-intervals
are considered). To illustrate this beneficial property, we also formulate
and solve such an optimization problem for the considered sampled-data
model of the gantry crane.

The paper is organized as follows: In Section \ref{sec:Flatness-and-Discretization}
we recapitulate the concept of continuous-time as well as discrete-time
flatness. We present a constructive method to derive a flat discretization
and illustrate it with the gantry crane. Next, in Section \ref{sec:Discrete-time-Flatness-based-Tra},
we use the sampled-data model to design tracking controllers based
on a linearization by a dynamic and a quasi-static state feedback.
Furthermore, we show how the controllers can be extended by integral
parts to compensate for stationary errors caused by friction effects.
In Section \ref{sec:Trajectory-Planning-and}, we present measurement
results and compare the novel discrete-time quasi-static controller
with the continuous-time controller derived in \cite{KolarRamsSchlacher:2017}
for different sampling rates. Moreover, we calculate an optimal trajectory
which achieves a transition between two given rest positions while
minimizing the acceleration of the load. Finally, we present measurement
results for the optimized trajectory as well as for the tracking of
a closed path in the shape of a lying eight.

\section{Flatness and Discretization\label{sec:Flatness-and-Discretization}}

In this section, first we briefly recall the definition of flatness
for continuous-time systems as well as a recent definition of (forward-)flatness
for discrete-time systems, see e.g. \cite{DiwoldKolarSchoeberl2022},
\cite{KolarSchoberlDiwold:2019}, or \cite{KolarKaldmaeSchoberlKottaSchlacher:2016}.
Subsequently, we show that for flat continuous-time systems which
can be transformed into a structurally flat triangular representation,
it is always possible to derive a forward-flat discretization that
can be used for the flatness-based controller design. We illustrate
this approach with the practical example of the gantry crane and derive
a forward-flat sampled-data model. This discrete-time system serves
as a basis for the controller design discussed in the remainder of
the paper.

\subsection{Flatness of continuous-time and discrete-time systems}

A continuous-time system in state representation
\begin{equation}
\dot{x}=f(x,u)\label{eq:cont_sys_equation}
\end{equation}
with $\dim(x)=n$ and $\dim(u)=m$ is flat, if there exists an $m$-tuple
$y=\varphi(x,u,\dot{u},\dots,u^{(q)})$ and a multi-index $r=(r_{1},\dots,r_{m})$
such that locally
\begin{align}
\begin{aligned}x & =F_{x}(y,\dot{y},\dots,y^{(r-1)})\\
u & =F_{u}(y,\dot{y},\dots,y^{(r)})\,,
\end{aligned}
\label{eq:cont_parameterizing_map_gen}
\end{align}
i.e., $x$ and $u$ can be expressed by $y$ and its time derivatives.
The $m$-tuple $y=\varphi(x,u,\dot{u},\dots,u^{(q)})$ is a flat output.
The map $F=(F_{x},F_{u})$ is necessarily a submersion and is denoted
in the following as parameterizing map. From this map it can be observed
that the trajectories $(x(t),u(t))$ of \eqref{eq:cont_sys_equation}
can be parameterized by sufficiently often differentiable trajectories
of the flat output $y(t)$. More precisely, there exists a one-to-one
correspondence between trajectories of \eqref{eq:cont_sys_equation}
and (sufficiently differentiable) arbitrary trajectories $y(t)$ of
a trivial system.

This idea can be transferred to the discrete-time framework, where
time derivatives have to be replaced by shifts. In \cite{DiwoldKolarSchoeberl2022},
we have shown that in general both forward- and backward-shifts can
be taken into account. In the following, however, we restrict ourselves
to the special case of forward-flat systems, since it is sufficient
for the practical application considered in this paper. We call a
discrete-time system in state representation
\begin{equation}
x^{+}=f(x,u)\label{eq:disc_sys_equation}
\end{equation}
with $\dim(x)=n$, $\dim(u)=m$ and $\textrm{rank}(\partial_{(x,u)}f)=n$
forward-flat, if there exists an $m$-tuple $y=\varphi(x,u,u_{[1]},\dots,u_{[q]})$
and a multi-index $r=(r_{1},\dots,r_{m})$ such that locally
\begin{align}
\begin{aligned}x & =F_{x}(y,y_{[1]},\dots,y_{[r-1]})\\
u & =F_{u}(y,y_{[1]},\dots,y_{[r]})\,,
\end{aligned}
\label{eq:discrete_parameterizing_map_general}
\end{align}
i.e., $x$ and $u$ can be expressed by $y$ and its forward-shifts.
Again, the map $F=(F_{x},F_{u})$ is necessarily a submersion and
allows to parameterize the trajectories $(x(k),u(k))$ of \eqref{eq:disc_sys_equation}
by arbitrary trajectories $y(k)$ of the flat output. In contrast
to the continuous-time case, the sequence $y(k)$ does not need to
be differentiable.
\begin{rem}
	Both in the continuous-time and the discrete-time case, checking whether
	a system is flat or not is a challenging task. The reason is that
	there exist no systematically verifiable necessary and sufficient
	conditions. For forward-flatness of discrete-time systems, however,
	verifiable conditions have been derived in \cite{KolarSchoberlDiwold:2019}.
	These conditions even lead to the computationally efficient test for
	forward-flatness presented in \cite{KolarDiwoldSchoberl:2019}, which
	is a generalization of the test for static feedback linearizability
	(see e.g. \cite{NijmeijervanderSchaft:1990}).
\end{rem}

\subsection{\label{subsec:Discretization-of-flat}Discretization of flat continuous-time
	systems}

In order to design a discrete-time flatness-based controller, a sampled-data
model of the plant has to be derived. To capture the effect that in
digital control circuits the control input is constant during each
sampling interval, preferably an exact discretization is used. However,
for most nonlinear systems determining an exact discretization is
rather difficult. Furthermore, in general it does not preserve the
flatness of a system, see e.g. \cite{DiwoldKolarSchoeberl2022}. From
a practical point of view, using the explicit\footnote{We use the supplement \textquotedbl explicit\textquotedbl{} to distinguish
	it from the implicit Euler-discretization. Since within this contribution
	we consider only explicit Euler-discretizations, we omit the supplement
	in the following.} Euler-discretization
\begin{equation}
x^{+}=x+T_{s}f(x,u)\label{eq:euler_disk_equations}
\end{equation}
is much more convenient. However, even the Euler-discretization does
not necessarily preserve flatness. Hence, in the following we show
that it may be useful to perform a suitable state transformation before
the Euler-discretization \eqref{eq:euler_disk_equations}. In other
words, we exploit the fact that state transformations and the Euler-discretization
do not commute in general. Input transformations and the Euler-discretization,
however, do commute, as the following diagram shows:
\[
\begin{array}{ccc}
\dot{x}=f(x,u) & \rightarrow & \begin{aligned}x^{+}=x+T_{s}f(x,u)\end{aligned}
\\
\downarrow &  & \downarrow\\
\dot{x}=f(x,\Phi_{u}(x,\bar{u})) & \rightarrow & \begin{aligned}x^{+}=x+T_{s}f(x,\Phi_{u}(x,\bar{u}))\end{aligned}
\end{array}
\]
Hence, the Euler-discretizations of continuous-time systems that are
related by an input transformation are again related by the same input
transformation. Since we make use of this fact in the sequel, we state
the following lemma.
\begin{lem}
	\label{lem:Input_transf_and_euler_disk_commute}Input transformations
	$u=\Phi_{u}(x,\bar{u})$ and the Euler-discretization \eqref{eq:euler_disk_equations}
	commute.
\end{lem}

\begin{rem}
	\label{rem:input_transf_disk_do_not_commute}For an exact discretization
	it is just vice versa: state transformations and exact discretization
	commute, whereas input transformations and exact discretization do
	in general not commute. If we introduce $\bar{u}$ via a feedback
	$u=\Phi_{u}(x,\bar{u})$ which explicitly depends on $x$, it is clear
	that for piecewise constant $u$ the new input $\bar{u}$ will in
	general not be constant over a sampling interval. Thus, an exact discretization
	with piecewise constant $\bar{u}$ will lead to a non-equivalent discrete-time
	system.
\end{rem}

Before we state the main result of this section, we introduce a certain
structurally flat triangular system representation that has two useful
properties. First, it allows to determine the parameterizing map \eqref{eq:cont_parameterizing_map_gen}
in a straightforward way. Second, with such a representation the property
of flatness is preserved by an Euler-discretization, as we show below.
\begin{defn}
	A system representation
	\begin{align}
	\begin{aligned}\dot{\bar{x}}_{p} & =f_{p}(\bar{x}_{p},\bar{x}_{p-1},\bar{u}_{p-1})\\
	\dot{\bar{x}}_{p-1} & =f_{p-1}(\bar{x}_{p},\dots,\bar{x}_{p-2},\bar{u}_{p-1},\bar{u}_{p-2})\\
	& \vdots\\
	\dot{\bar{x}}_{2} & =f_{2}(\bar{x}_{p},\dots,\bar{x}_{1},\bar{u}_{p-1},\dots,\bar{u}_{1})\\
	\dot{\bar{x}}_{1} & =f_{1}(\bar{x}_{p},\dots,\bar{x}_{1},\bar{u}_{p-1},\dots,\bar{u}_{1},\bar{u}_{0})
	\end{aligned}
	\label{eq:struct_triang_rep}
	\end{align}
	with the state $\bar{x}=(\bar{x}_{p},\bar{x}_{p-1}\dots,\bar{x}_{1})=(y_{p},(y_{p-1},\hat{x}_{p-1}),\dots,(y_{1},\hat{x}_{1}))$
	and the input $\bar{u}=(\bar{u}_{p-1},\dots,\bar{u}_{0})$ (with possibly
	empty components $y_{l}$ and/or $\bar{u}_{l}$) that satisfies the
	rank conditions
	\begin{align}
	\begin{aligned}\mathrm{rank}(\partial_{\bar{u}_{0}}f_{1}) & =\dim(f_{1})=\dim(\bar{u}_{0})\\
	\textrm{\ensuremath{\mathrm{rank}}}(\partial_{(\hat{x}_{l},\bar{u}_{l})}f_{l+1}) & =\dim(f_{l+1})=\dim(\hat{x}_{l})+\dim(\bar{u}_{l})
	\end{aligned}
	\label{eq:rank_cond_triangular}
	\end{align}
	for $l=1,\dots,p-1$ is denoted in the following as structurally flat
	triangular form. A flat output is given by $y=(y_{p},\dots,y_{1})$.
\end{defn}

The rank conditions \eqref{eq:rank_cond_triangular} guarantee that
all states and inputs can be expressed as functions of the flat output
$y$ and its time derivatives by solving the equations \eqref{eq:struct_triang_rep}
from top to bottom. Since the Euler-discretization of a continuous-time
system \eqref{eq:struct_triang_rep} has the same triangular structure,
we immediately get the following theorem.
\begin{thm}
	\label{thm:Exp_euler_of_triangular_rep_is_flat}The Euler-discretization
	of a continuous-time system in structurally flat triangular form \eqref{eq:struct_triang_rep}
	is forward-flat, and the flat output $y=(y_{p},\dots,y_{1})$ is preserved.
\end{thm}

\begin{proof}
	The Euler-discretization of \eqref{eq:struct_triang_rep} is given
	by
	\begin{align}
	\begin{aligned}\bar{x}_{p}^{+} & =\bar{x}_{p}+T_{s}f_{p}(\bar{x}_{p},\bar{x}_{p-1},\bar{u}_{p-1})\\
	\bar{x}_{p-1}^{+} & =\bar{x}_{p-1}+T_{s}f_{p-1}(\bar{x}_{p},\dots,\bar{x}_{p-2},\bar{u}_{p-1},\bar{u}_{p-2})\\
	& \vdots\\
	\bar{x}_{2}^{+} & =\bar{x}_{2}+T_{s}f_{2}(\bar{x}_{p},\dots,\bar{x}_{1},\bar{u}_{p-1},\dots,\bar{u}_{1})\\
	\bar{x}_{1}^{+} & =\bar{x}_{1}+T_{s}f_{1}(\bar{x}_{p},\dots,\bar{x}_{1},\bar{u}_{p-1},\dots,\bar{u}_{1},\bar{u}_{0})\,,
	\end{aligned}
	\label{eq:struct_triang_rep_disc}
	\end{align}
	and due to the rank conditions \eqref{eq:rank_cond_triangular} all
	states $\bar{x}$ and inputs $\bar{u}$ can be expressed as functions
	of the flat output $y$ and its forward-shifts by solving the equations
	\eqref{eq:struct_triang_rep_disc} from top to bottom. Thus, the discrete-time
	system \eqref{eq:struct_triang_rep_disc} is forward-flat with the
	flat output $y=(y_{p},\dots,y_{1})$.
\end{proof}
\begin{rem}
	The structurally flat triangular form \eqref{eq:struct_triang_rep}
	is actually more restrictive than necessary. With a more complex notation,
	it would be possible to state a more general version of \eqref{eq:struct_triang_rep}
	that still possesses a forward-flat Euler-discretization. Alternative
	flat normal forms can be found e.g. in \cite{GstoettnerKolarSchoeberl2021},
	\cite{GstoettnerKolarSchoeberl2020}, \cite{NicolauRespondek:2019}
	or \cite{NicolauRespondek:2020}. Furthermore, it should be noted
	that the flat output of \eqref{eq:struct_triang_rep} depends only
	on state variables. Thus, only x-flat systems can possess a triangular
	form \eqref{eq:struct_triang_rep}.
\end{rem}

If it is possible to transform a system \eqref{eq:cont_sys_equation}
by state- and input transformations\begin{subequations}
	\label{state_input_transformation_sub}
	\begin{align}
	\bar{x}&=\Phi_x(x)
	\label{state_transformation}\\
	\bar{u}&=\Phi_u(x,u)
	\label{input_transformation}
	\end{align}
\end{subequations}into the structurally flat triangular form \eqref{eq:struct_triang_rep},
then a combination of this transformation and a subsequent Euler-discretization
apparently yields a forward-flat discrete-time system. Nevertheless,
a discretization where the original input $u$ is preserved would
be favorable. This is indeed possible, since by Lemma \ref{lem:Input_transf_and_euler_disk_commute}
the structurally flat discrete-time triangular form \eqref{eq:struct_triang_rep_disc}
can also be obtained by discretizing the system \eqref{eq:cont_sys_equation}
after the state transformation \eqref{state_transformation} and performing
the input transformation \eqref{input_transformation} afterwards.
Since an input transformation does not affect the flatness of a system,
we get the following corollary.
\begin{cor}
	\label{cor:If-a-flat}If a flat continuous-time system can be transformed
	into a structurally flat triangular form \eqref{eq:struct_triang_rep}
	by a transformation \eqref{state_input_transformation_sub}, then
	the Euler-discretization $\bar{x}^{+}=\bar{x}+T_{s}f(\bar{x},u)$
	of the system $\dot{\bar{x}}=f(\bar{x},u)$ obtained by the state
	transformation \eqref{state_transformation} is forward-flat.
\end{cor}

We want to mention that the choice of a state is not unique, and hence
it is justified to perform a suitable state transformation prior to
the discretization. Nevertheless, it should be noted that not every
flat continuous-time system possesses a triangular representation
\eqref{eq:struct_triang_rep}. However, for many practical examples
like e.g. the gantry crane \cite{KolarRamsSchlacher:2017}, the VTOL
aircraft \cite{FliessLevineMartinRouchon:1999}, or the induction
motor \cite{Chiasson:1998} such a representation exists.

\subsection{\label{subsec:Forward-flat-sampled-data-model}Forward-flat sampled-data
	model of a gantry crane}

\begin{figure}
	\centering\includegraphics[width=0.5\columnwidth]{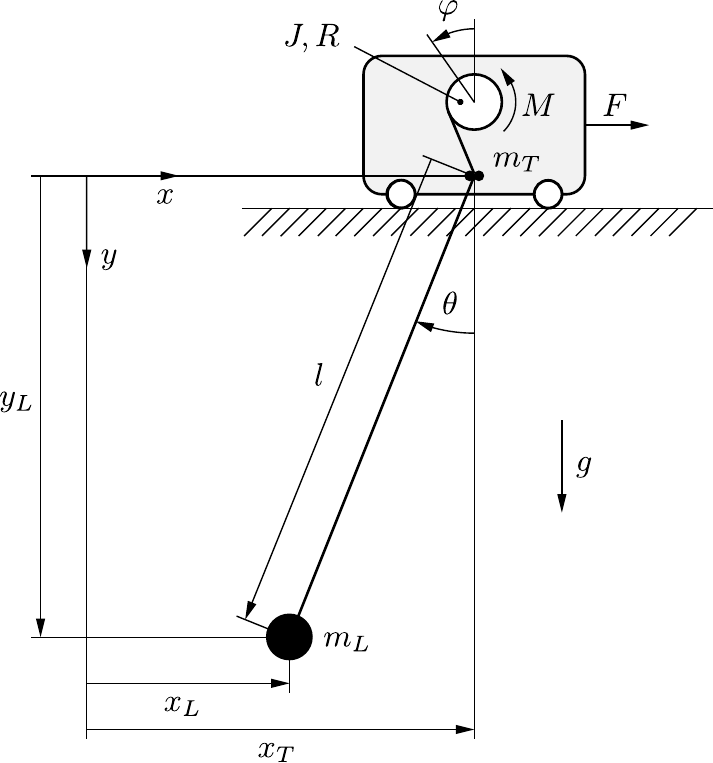}
	
	\caption{\label{fig:Schematic}Schematic diagram of the gantry crane.}
\end{figure}

In Fig. \ref{fig:Schematic}, the schematic model of the considered
gantry crane is shown. The position of the trolley is denoted by $x_{T}$,
the rotation angle of the rope drum is denoted by $\varphi$, and
$\theta$ describes the pendulum angle. With the length of the rope
$l=R\varphi$, the position of the load is given by $x_{L}=x_{T}-R\varphi\sin(\theta)$
and $y_{L}=R\varphi\cos(\theta)$. The mass of the trolley is denoted
by $m_{T}$, and $J$ represents the moment of inertia of the rope
drum. The load has the mass $m_{L}$, and the gravitational acceleration
$g$ points in the positive $y$-direction. The driving force $F$
which acts on the trolley and the driving torque $M$ which acts on
the rope drum are the control inputs. The equations of motion can
be found e.g. in \cite{KolarRamsSchlacher:2017} and read as
\begin{align*}
\begin{aligned}(m_{T}+m_{L})\ddot{x}_{T}-m_{L}R\sin(\theta)\ddot{\varphi}-m_{L}R\varphi\cos(\theta)\ddot{\theta}\;\\
+m_{L}\dot{\theta}(R\varphi\dot{\theta}\sin(\theta)-2R\dot{\varphi}\cos(\theta)) & =F\,,\\
-m_{L}R\sin(\theta)\ddot{x}_{T}+(J+m_{L}R^{2})\ddot{\varphi}\;\\
-m_{L}R(R\varphi\dot{\theta}^{2}+g\cos(\theta)) & =M\,,\\
-\cos(\theta)\ddot{x}_{T}+R\varphi\ddot{\theta}+2R\dot{\varphi}\dot{\theta}+g\sin(\theta) & =0\,.
\end{aligned}
\end{align*}
A state representation \eqref{eq:cont_sys_equation} is given by
\begin{equation}
\begin{aligned}\begin{aligned}\dot{x}_{T} & =v_{T}\\
\dot{\varphi} & =\omega_{\varphi}\\
\dot{\theta} & =\omega_{\theta}\\
\text{\ensuremath{\dot{v}_{T}}} & =f_{a_{T}}(\varphi,\theta,\omega_{\theta},F,M)\\
\ensuremath{\dot{\omega}_{\varphi}} & =f_{\alpha_{\varphi}}(\varphi,\theta,\omega_{\theta},F,M)\\
\ensuremath{\dot{\omega}_{\theta}} & =f_{\alpha_{\theta}}(\varphi,\theta,\omega_{\varphi},\omega_{\theta},F,M)
\end{aligned}
\end{aligned}
\label{eq:gantry_original_coordinates}
\end{equation}
with the state $x=(x_{T},\varphi,\theta,v_{T},\omega_{\varphi},\omega_{\theta})$
and the input $u=(F,M)$. It is well-known that the model of the gantry
crane is flat and that the position of the load is a flat output $y=(x_{T}-R\varphi\sin(\theta),R\varphi\cos(\theta))$.
In order to derive a sampled-data model, in a first attempt one may
simply compute the Euler-discretization of \eqref{eq:gantry_original_coordinates}.
However, it can be shown that the resulting discrete-time system does
not meet the necessary conditions given in \cite{KolarSchoberlDiwold:2019}
and is thus not forward-flat. In fact, the test for forward-flatness
stated in \cite{KolarDiwoldSchoberl:2019} fails already in the second
step. Thus, we follow the approach suggested in Section \ref{subsec:Discretization-of-flat}
and search for a state- and input transformation \eqref{state_input_transformation_sub}
that transforms \eqref{eq:gantry_original_coordinates} into a structurally
flat triangular representation \eqref{eq:struct_triang_rep}. It turns
out that for the gantry crane such a transformation indeed exists.
A suitable state transformation is given by
\begin{equation}
\begin{aligned}\begin{aligned}x_{L} & =x_{T}-R\varphi\sin(\theta)\\
y_{L} & =R\varphi\cos(\theta)\\
v_{L,x} & =v_{T}-\omega_{\varphi}R\sin(\theta)-\omega_{\theta}R\varphi\cos(\theta)\\
v_{L,y} & =R(\omega_{\varphi}\cos(\theta)-\sin(\theta)\varphi\omega_{\theta})\,,
\end{aligned}
\end{aligned}
\label{eq:coord_transf}
\end{equation}
where the transformed state $\bar{x}=(x_{L},y_{L},v_{L,x},v_{L,y},\theta,\omega_{\theta})$
contains the flat output and its first time derivative. In such coordinates,
the system equations read as
\begin{equation}
\begin{aligned}\begin{aligned}\dot{x}_{L} & =v_{L,x}\\
\dot{y}_{L} & =v_{L,y}\\
\text{\ensuremath{\dot{\theta}}} & =\omega_{\theta}\\
\dot{v}_{L,x} & =f_{a_{L,x}}(y_{L},\theta,\omega_{\theta},M)\\
\dot{v}_{L,y} & =f_{a_{L,y}}(y_{L},\theta,\omega_{\theta},M)\\
\ensuremath{\dot{\omega}_{\theta}} & =f_{\alpha_{\theta}}(v_{L,y},\theta,\omega_{\theta},F,M)\,.
\end{aligned}
\end{aligned}
\label{eq:ganry_new_state}
\end{equation}
If we additionally introduce the new input $\bar{u}=(a_{L,x},\alpha_{\theta})$
defined by
\begin{align*}
\begin{aligned}\begin{aligned}a_{L,x} & =f_{a_{L,x}}(y_{L},\theta,\omega_{\theta},M)\\
\alpha_{\theta} & =f_{\alpha_{\theta}}(v_{L,y},\theta,\omega_{\theta},F,M)\,,
\end{aligned}
\end{aligned}
\end{align*}
i.e., the acceleration $a_{L,x}$ of the load in $x$-direction and
the angular acceleration $\alpha_{\theta}$, we obtain a structurally
flat triangular representation
\begin{equation}
\begin{aligned}\dot{x}_{L} & =v_{L,x}\\
\dot{y}_{L} & =v_{L,y}\\
\dot{v}_{L,x} & =a_{L,x}\\
\dot{v}_{L,y} & =g-\tfrac{a_{L,x}}{\tan(\theta)}\\
\text{\ensuremath{\dot{\theta}}} & =\omega_{\theta}\\
\ensuremath{\dot{\omega}_{\theta}} & =\alpha_{\theta}
\end{aligned}
\label{eq:gantry_struct_triang_rep}
\end{equation}
with the flat output $y=(x_{L},y_{L})$. Due to the triangular structure
the Euler-discretization of \eqref{eq:gantry_struct_triang_rep} is
forward-flat (see Theorem \ref{thm:Exp_euler_of_triangular_rep_is_flat}),
and because of Corollary \ref{cor:If-a-flat} the same applies to
the Euler-discretization of \eqref{eq:ganry_new_state} with the original
inputs $(F,M)$ given by
\begin{equation}
\begin{aligned}\begin{aligned}x_{L}^{+} & =x_{L}+T_{s}v_{L,x}\\
y_{L}^{+} & =y_{L}+T_{s}v_{L,y}\\
\text{\ensuremath{\theta^{+}}} & =\theta+T_{s}\omega_{\theta}\\
v_{L,x}^{+} & =v_{L,x}+T_{s}f_{a_{L,x}}(y_{L},\theta,\omega_{\theta},M)\\
v_{L,y}^{+} & =v_{L,y}+T_{s}f_{a_{L,y}}(y_{L},\theta,\omega_{\theta},M)\\
\ensuremath{\omega_{\theta}^{+}} & =\omega_{\theta}+T_{s}f_{\alpha_{\theta}}(v_{L,y},\theta,\omega_{\theta},F,M)\,.
\end{aligned}
\end{aligned}
\label{eq:sampled_data_model_gantry_crane}
\end{equation}
The corresponding parameterizing map \eqref{eq:discrete_parameterizing_map_general}
is of the form
\begin{align}
\begin{aligned}\bar{x} & =F_{\bar{x}}(x_{L},y_{L},\dots,x_{L,[3]},y_{L,[3]})\\
u & =F_{u}(x_{L},y_{L},\dots,x_{L,[4]},y_{L,[4]})
\end{aligned}
\label{eq:parameterizing_map_gantry}
\end{align}
and will be used in the next section for the design of two different
tracking controllers. For simplicity, in the remainder of the paper
the bar-notation will be omitted.

\section{Discrete-time Flatness-based Tracking Control\label{sec:Discrete-time-Flatness-based-Tra}}

The design of a continuous-time flatness-based tracking control is
typically divided into two steps. First, the system is exactly linearized
by either an endogenous dynamic feedback or a quasi-static state feedback,
see e.g. \cite{FliessLevineMartinRouchon:1995} and \cite{DelaleauRudolph:1998}.
Subsequently, a further feedback is applied in order to stabilize
reference trajectories. However, since the resulting control law is
assumed to be evaluated continuously, it actually cannot be realized
on a digital computer/processor with only a finite sampling rate.
Nevertheless, for applications where the sampling rate can be chosen
sufficiently high, a discrete-time evaluation of the control law usually
leads to satisfying results.

In this section, however, we propose an alternative approach and design
flatness-based tracking controllers directly for the sampled-data
model \eqref{eq:sampled_data_model_gantry_crane} of the gantry crane.
Like in the continuous-time case, we follow the two-step design philosophy
with an exact linearization and a subsequent stabilizing feedback.
For the exact linearization we use two different approaches based
on a dynamic as well as a quasi-static state feedback. Furthermore,
we demonstrate how the controllers can be extended by integral parts
in order to compensate for stationary deviations caused e.g. by model
uncertainties.

\subsection{\label{subsec:Design_of_dyn_feedback}Tracking control with linearization
	by endogenous dynamic feedback}

In \cite{DiwoldKolarSchoeberl2022}, we have shown that discrete-time
flat systems (including both forward- and backward shifts) can be
linearized by a certain class of dynamic feedback
\begin{align*}
\begin{aligned}z^{+} & =\alpha(x,z,v)\\
u & =\beta(x,z,v)
\end{aligned}
\end{align*}
that closely mimics the class of endogenous dynamic feedback from
the continuous-time case. In the following, we demonstrate the construction
of such a linearizing dynamic feedback as well as an additional feedback
to stabilize reference trajectories.

First, we extend the parameterizing map \eqref{eq:discrete_parameterizing_map_general}
by a map $z=F_{z}(y,\dots,y_{[r-1]})$ with $\dim(z)=\sum_{j=1}^{m}r_{j}-n$,
such that the combined map $F_{xz}=(F_{x},F_{z})$ has a regular Jacobian
matrix and is hence invertible. Since $F_{x}$ is a submersion, such
an extension always exists. Next, we define the map $\Phi(y,\dots,y_{[r]})$
given by
\begin{equation}
\begin{aligned}x & =F_{x}(y,\dots,y_{[r-1]})\\
z & =F_{z}(y,\dots,y_{[r-1]})\\
v & =y_{[r]}
\end{aligned}
\label{eq:phi}
\end{equation}
and its inverse $\hat{\Phi}(x,z,v)$ given by 
\begin{align}
\begin{array}{c}
\begin{aligned}(y,\dots,y_{[r-1]}) & =\hat{F}_{xz}(x,z)\\
y_{[r]} & =v\,.
\end{aligned}
\end{array}\label{eq:phi_d}
\end{align}
As shown in \cite{DiwoldKolarSchoeberl2022}, a linearizing dynamic
feedback follows by evaluating\footnote{Note that we use the shift operator $\delta_{y}$ defined in \cite{DiwoldKolarSchoeberl2022},
	which acts in the special case of forward-flatness on a function $H$
	according to the rule $\delta_{y}(H(y,y_{[1]},\dots))=H(y_{[1]},y_{[2]},\dots)$.}
\begin{align}
\begin{aligned}z^{+} & =\delta_{y}(F_{z})\circ\hat{\Phi}(x,z,v)\\
u & =F_{u}\circ\hat{\Phi}(x,z,v)\,.
\end{aligned}
\label{eq:dynamic_feedback}
\end{align}
The state- and input transformation \eqref{eq:phi_d} would transform
the closed-loop
\begin{align*}
\begin{aligned}x^{+} & =f(x,F_{u}\circ\hat{\Phi}(x,z,v))\\
z^{+} & =\delta_{y}(F_{z})\circ\hat{\Phi}(x,z,v)
\end{aligned}
\end{align*}
into Brunovsky normal form $y_{[r_{j}]}^{j}=v^{j}$ for $j=1,\dots,m$.
In such coordinates, a tracking controller for a reference trajectory
$y_{d}(k)$ can now be designed easily. For this purpose, we introduce
the tracking error and its forward-shifts as
\begin{equation}
e_{[i]}^{j}=y_{[i]}^{j}-y_{d,[i]}^{j}\,,\hphantom{a}j=1,\dots,m\,.\label{eq:def_err}
\end{equation}
In order to stabilize a reference trajectory we use a linear tracking
error dynamics of the form
\begin{equation}
e_{[r_{j}]}^{j}+\sum_{i=0}^{r_{j}-1}a_{i}^{j}e_{[i]}^{j}=0\,,\hphantom{a}j=1,\dots,m\,,\label{eq:error_dynamic}
\end{equation}
where the coefficients $a_{i}^{j}$ are chosen such that the roots
of \eqref{eq:error_dynamic} are located within the unit circle. By
substituting the definition \eqref{eq:def_err} of the tracking error
into \eqref{eq:error_dynamic} and solving for $y_{[r_{j}]}=v$, the
stabilizing feedback for the new input $v$ follows as
\begin{equation}
v^{j}=y_{d,[r_{j}]}^{j}-\sum_{i=0}^{r_{j}-1}a_{i}^{j}(y_{[i]}^{j}-y_{d,[i]}^{j})\,,\hphantom{a}j=1,\dots,m\,.\label{eq:feedback_law}
\end{equation}
In a last step, we have to express the forward-shifts of the flat
output that appear in \eqref{eq:feedback_law} by the system state
$x$ and the controller state $z$. By substituting $(y,\dots,y_{[r-1]})=\hat{F}_{xz}(x,z)$
according to the map \eqref{eq:phi_d} into \eqref{eq:feedback_law}
and subsequently \eqref{eq:feedback_law} into \eqref{eq:dynamic_feedback},
the complete control law follows as
\begin{align*}
\begin{aligned}z^{+} & =\zeta(x,z,y_{d},\dots,y_{d,[r]})\\
u & =\eta(x,z,y_{d},\dots,y_{d,[r]})\,.
\end{aligned}
\end{align*}
It depends on the state $x$ of the system \eqref{eq:disc_sys_equation},
the controller state $z$, and the reference trajectory $y_{d}$ as
well as its forward-shifts. Thus, for a practical implementation,
the state $x$ has to be measured or estimated by an observer.

For the sampled-data model \eqref{eq:sampled_data_model_gantry_crane}
of the gantry crane, a linearizing dynamic feedback \eqref{eq:dynamic_feedback}
can be constructed by choosing $F_{z}$ of \eqref{eq:phi} as\footnote{Note that the choices $z^{1}=x_{L[2]}$, $z^{2}=x_{L[3]}$ and $z^{1}=x_{L[3]}$,
	$z^{2}=y_{L[3]}$ would also be possible. However, the corresponding
	dynamic feedbacks \eqref{eq:dynamic_feedback} have singularities
	at $\theta=\omega_{\theta}=0$ and are hence not useful for a practical
	application.}
\begin{align*}
z^{1} & =y_{L,[2]}\\
z^{2} & =y_{L,[3]}\,.
\end{align*}
The order of the resulting tracking error dynamics \eqref{eq:error_dynamic}
is given by $r=(4,4)$. As we show in the next section, a lower order
can be achieved by a linearization with a quasi-static feedback.

\subsection{\label{subsec:Design_of_quasi-static}Tracking control with linearization
	by quasi-static state feedback}

In \cite{DelaleauRudolph:1998}, it is shown that every flat continuous-time
system can be linearized by a quasi-static state feedback. In contrast
to the linearization by dynamic feedback, the resulting control law
is static but involves time derivatives of the closed-loop input.
Despite these time derivatives of the new input, the approach is very
useful for the design of flatness-based tracking controls. The main
advantage over designs with dynamic feedback linearization is that
it leads to lower order tracking error dynamics. In the following,
we transfer this idea to discrete-time systems and design a tracking
control based on a linearization by a discrete-time quasi-static state
feedback for the sampled-data model \eqref{eq:sampled_data_model_gantry_crane}
of the gantry crane.

The exact linearization of flat systems by quasi-static feedback discussed
in \cite{DelaleauRudolph:1998} is based on the concept of generalized
state representations. According to \cite{Kotta1998}, for discrete-time
systems generalized state representations are of the form
\[
\tilde{x}^{+}=f(\tilde{x},u,u_{[1]},\dots,u_{[\gamma]})
\]
and can also depend on forward-shifts of the input (instead of time
derivatives in the continuous-time case). Such a generalized system
representation follows from \eqref{eq:disc_sys_equation} by a generalized
state transformation of the form
\begin{equation}
\tilde{x}=\Psi_{x}(x,u,u_{[1]},\dots,u_{[\nu]})\label{eq:generalized_state_transformation}
\end{equation}
with $\textrm{rank}(\partial_{x}\Psi_{x})=n$. The rank condition
ensures that the generalized state transformation \eqref{eq:generalized_state_transformation}
is invertible. For the class of forward-flat systems, it can be shown
that there always exist generalized states of the form
\begin{equation}
\tilde{x}_{B}=(y^{1},y_{[1]}^{1},\dots,y_{[\kappa_{1}-1]},\dots,y^{m},y_{[1]}^{m},\dots,y_{[\kappa_{m}-1]}^{m})\label{eq:gen_brunvosky_state}
\end{equation}
which consist of the components of the flat output and its forward-shifts
up to some order $\kappa-1=(\kappa_{1}-1,\dots,\kappa_{m}-1)$ with
$\dim(\tilde{x}_{B})=n$ and $\kappa_{j}\leq r_{j}$. In accordance
with the continuous-time case \cite{DelaleauRudolph:1998}, we call
\eqref{eq:gen_brunvosky_state} a generalized Brunovsky state. With
a feedback of the form
\begin{equation}
u=F_{u}(\tilde{x}_{B},v,v_{[1]},\dots,v_{[r-\kappa]})\,,\label{eq:quasi_static_linearizing_feedback}
\end{equation}
which is constructed by replacing the $n$ components of $y,\dots,y_{[\kappa-1]}$
in the map $F_{u}$ of \eqref{eq:discrete_parameterizing_map_general}
with the Brunovsky state $\tilde{x}_{B}$ and the remaining components
$y_{[\kappa]},\dots,y_{[r]}$ with the new input and its forward-shifts
$v,\dots,v_{[r-\kappa]}$, we get an input/output behaviour of the
form 
\begin{equation}
y_{[\kappa_{j}]}^{j}=v^{j}\,,\hphantom{a}j=1,\dots,m\,.\label{eq:brunv_quasi_static}
\end{equation}
In the continuous-time case discussed in \cite{DelaleauRudolph:1998},
such a feedback is called a quasi-static state feedback\footnote{For discrete-time systems, quasi-static feedbacks are used e.g. in
	\cite{ArandaKotta2001}.}. With the linear input/output behaviour \eqref{eq:brunv_quasi_static},
it is now easy to construct a further feedback in order to introduce
e.g. a linear tracking error dynamics
\begin{equation}
e_{[\kappa_{j}]}^{j}+\sum_{i=0}^{\kappa_{j}-1}a_{i}^{j}e_{[i]}^{j}=0\,,\hphantom{a}j=1,\dots,m\,.\label{eq:error_dynamic_quasi_static}
\end{equation}
The coefficients $a_{i}^{j}$ must be chosen such that the roots of
\eqref{eq:error_dynamic_quasi_static} are located within the unit
circle. By substituting the definition \eqref{eq:def_err} of the
tracking error into \eqref{eq:error_dynamic_quasi_static} and solving
for $y_{[\kappa_{j}]}=v$, the stabilizing feedback follows as
\[
v^{j}=y_{d,[\kappa_{j}]}^{j}-\sum_{i=0}^{\kappa_{j}-1}a_{i}^{j}(y_{[i]}^{j}-y_{d,[i]}^{j})\,,\hphantom{a}j=1,\dots,m\,.
\]
The forward-shifts of $v$ that are required for the quasi-static
feedback \eqref{eq:quasi_static_linearizing_feedback} follow as
\[
v_{[p]}^{j}=y_{d,[\kappa_{j}+p]}^{j}-\sum_{i=0}^{\kappa_{j}-1}a_{i}^{j}(y_{[i+p]}^{j}-y_{d,[i+p]}^{j})
\]
for $p=1,\dots,r_{j}-\kappa_{j}$. With the relation $v_{[p]}^{j}=y_{[\kappa_{j}+p]}^{j}$,
which is a consequence of \eqref{eq:brunv_quasi_static}, we get the
set of linear equations
\begin{align}
\begin{aligned}v^{j} & =y_{d,[\kappa_{j}]}^{j}-\sum_{i=0}^{\kappa_{j}-1}a_{i}^{j}(y_{[i]}^{j}-y_{d,[i]}^{j})\\
v_{[1]}^{j} & =y_{d,[\kappa_{j}+1]}^{j}-a_{\kappa_{j}-1}^{j}(v^{j}-y_{d,[\kappa_{j}]}^{j})-\sum_{i=0}^{\kappa_{j}-2}a_{i}^{j}(y_{[i+1]}^{j}-y_{d,[i+1]}^{j})\\
& \vdots\\
v_{[r_{j}-\kappa_{j}]}^{j} & =\ldots
\end{aligned}
\label{eq:set_of_equ_quasi_static}
\end{align}
$j=1,\dots,m$ for $v,\dots,v_{[r-\kappa]}$. One could now solve
\eqref{eq:set_of_equ_quasi_static} for $v,\dots,v_{[r-\kappa]}$
systematically from top to bottom and insert the solution into \eqref{eq:quasi_static_linearizing_feedback}.
However, the resulting control law would depend on the generalized
Brunovsky state \eqref{eq:gen_brunvosky_state}, i.e., on forward-shifts
of the flat output up to order $\kappa-1$. Since future values of
the flat output can in general not be measured, the resulting control
law is obviously difficult to implement. Thus, we try to replace the
flat output and its forward-shifts beforehand, which is possible if
two assumptions are met. First, we assume that with
\[
x=F_{x}(\tilde{x}_{B},v,v_{[1]},\dots,v_{[r-\kappa-1]})\,,
\]
where we have inserted \eqref{eq:gen_brunvosky_state} as well as
\eqref{eq:brunv_quasi_static} and its forward-shifts into the map
$F_{x}$ of \eqref{eq:discrete_parameterizing_map_general}, the condition
$\textrm{\ensuremath{\mathrm{rank}}}(\partial_{\tilde{x}_{B}}F_{x})=n$
is met. This guarantees by the implicit function theorem that the
Brunovsky state can be expressed as 
\begin{equation}
\tilde{x}_{B}=\phi(x,v,v_{[1]},\dots,v_{[r-\kappa-1]})\,.\label{eq:solved_for_x_B}
\end{equation}
Substituting \eqref{eq:solved_for_x_B} into \eqref{eq:quasi_static_linearizing_feedback}
yields the linearizing feedback
\begin{equation}
u=F_{u}(\phi(x,v,v_{[1]},\dots,v_{[r-\kappa-1]}),v,v_{[1]},\dots,v_{[r-\kappa]})\,,\label{eq:quasi_static_linearizing_feedback_exp_by_x}
\end{equation}
which introduces the same input/output behaviour \eqref{eq:brunv_quasi_static}
as \eqref{eq:quasi_static_linearizing_feedback} but depends on $x$
instead of $\tilde{x}_{B}$. However, we also have to determine $v,v_{[1]},\dots,v_{[r-\kappa]}$
as a function of $x$ instead of $\tilde{x}_{B}$. For this purpose,
we substitute \eqref{eq:solved_for_x_B} into \eqref{eq:set_of_equ_quasi_static}
and obtain a set of in general nonlinear equations. Thus, our second
assumption is that this set of nonlinear equations can nevertheless
be solved for $v,\dots,v_{[r-\kappa]}$, i.e.,
\begin{equation}
(v,v_{[1]},\dots,v_{[r-\kappa]})=\chi(x,y_{d},\dots,y_{d,[r]})\,.\label{eq:solved_for_v_v_1}
\end{equation}
Finally, by substituting \eqref{eq:solved_for_v_v_1} into \eqref{eq:quasi_static_linearizing_feedback_exp_by_x},
we obtain a static control law of the form
\begin{equation}
u=\eta(x,y_{d},\dots,y_{d,[r]})\,,\label{eq:quasi_static_final_control_law}
\end{equation}
which only depends on the state $x$ as well as the reference trajectory
$y_{d}$ and its forward-shifts. Thus, like for the dynamic tracking
controller of Section \ref{subsec:Design_of_dyn_feedback}, only the
state $x$ is required for an implementation.

For the sampled-data model \eqref{eq:sampled_data_model_gantry_crane}
of the gantry crane, a generalized Brunovsky state \eqref{eq:gen_brunvosky_state}
is given by\footnote{The choices $\tilde{x}_{B}=(x_{L},\dots,x_{L,[2]},y_{L},\dots,y_{L,[2]})$
	and $\tilde{x}_{B}=(x_{L},x_{L,[1]},y_{L},\dots,y_{L,3]})$ would
	also be possible. However, the corresponding control laws \eqref{eq:quasi_static_final_control_law}
	would again have singularities at $\theta=\omega_{\theta}=0$ and
	are not useful for a practical application.}
\[
\tilde{x}_{B}=(x_{L},\dots,x_{L,[3]},y_{L},y_{L,[1]})\,.
\]
Since both above assumptions are met, a control law of the form \eqref{eq:quasi_static_final_control_law}
can be derived. Compared to the dynamic controller of Section \ref{subsec:Design_of_dyn_feedback},
the order of the resulting tracking error dynamics \eqref{eq:error_dynamic_quasi_static}
is reduced to $\kappa=(4,2)$.

\subsection{\label{subsec:Extending-the-tracking}Extension of the tracking error
	dynamics by integral parts}

In this section, we show how the discrete-time tracking control with
dynamic feedback of Section \ref{subsec:Design_of_dyn_feedback} can
be extended by discrete-time integral parts in order to avoid stationary
errors caused e.g. by model uncertainties. In the case of the gantry
crane, such errors occur mainly due to friction effects. With the
``integrated'' tracking error $e_{I}$ given by
\[
e_{I}^{j,+}=e_{I}^{j}+T_{s}e^{j}\,,\hphantom{a}j=1,\dots,m\,,
\]
the extended error dynamics reads as
\[
\left[\begin{array}{c}
e_{I}^{j,+}\\
e^{j,+}\\
\vdots\\
e_{[r_{j}-1]}^{j,+}
\end{array}\right]=\underset{A^{j}}{\underbrace{\left[\begin{array}{ccccc}
		1 & T_{s} & 0 & 0 & 0\\
		0 & 0 & 1 & \ddots & 0\\
		\vdots & \vdots & \ddots & \ddots & 0\\
		0 & 0 & \cdots & 0 & 1\\
		-a_{I}^{j} & -a_{0}^{j} & \cdots & \cdots & -a_{r_{j}-1}^{j}
		\end{array}\right]}}\left[\begin{array}{c}
e_{I}^{j}\\
e^{j}\\
\vdots\\
e_{[r_{j}-1]}^{j}
\end{array}\right]\,,
\]
$j=1,\dots,m$. To obtain a stable error dynamics, the coefficients
$a_{I}^{j}$ and $a_{i}^{j}$ must be chosen such that all eigenvalues
of the matrices $A^{j}$ are located within the unit circle. The extended
stabilizing feedback for the exactly linearized system follows as
\[
\bar{v}^{j}=v^{j}-a_{I}^{j}e_{I}^{j}\,,\hphantom{a}j=1,\dots,m\,,
\]
where $v$ represents the feedback without integral parts obtained
in Section \ref{subsec:Design_of_dyn_feedback}. For the controller
with quasi-static state feedback of Section \ref{subsec:Design_of_quasi-static},
an extension by integral parts is possible in a similar way. With
respect to the gantry crane, the integral parts increase the order
of the error dynamics to $(5,5)$ for the dynamic controller and $(5,3)$
for the quasi-static controller.

\section{Trajectory Planning and Measurement Results\label{sec:Trajectory-Planning-and}}

\begin{figure}
	\includegraphics[width=1\columnwidth]{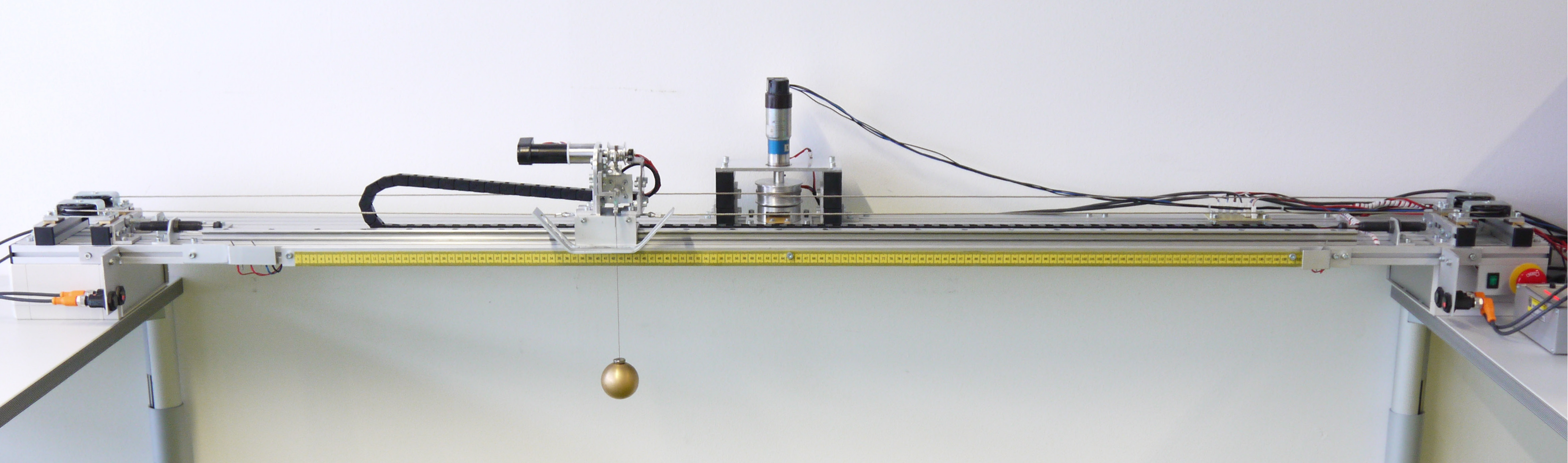}\centering\caption{\label{fig:Picture_gantry_crane}Laboratory model of the gantry crane.}
\end{figure}
In the following, we present measurement results for the novel discrete-time
quasi-static controller of Section \ref{subsec:Design_of_quasi-static},
which was implemented on a laboratory model of the gantry crane depicted
in Fig. \ref{fig:Picture_gantry_crane}. In particular, in Section
\ref{subsec:Comparison_cont_disc} we compare its tracking performance
for different sampling rates with the continuous-time quasi-static
tracking controller discussed in \cite{KolarRamsSchlacher:2017}.
Moreover, in Section \ref{subsec:Optimization} we show how the discrete-time
approach facilitates the design of optimal reference trajectories.
We formulate and solve an optimization problem and present measurement
results for the computed trajectories. Finally, in Section \ref{subsec:Lying_eight}
we show further measurement results for the tracking of a closed path.

Like in \cite{KolarRamsSchlacher:2017}, for all experiments the control
law \eqref{eq:quasi_static_final_control_law} is extended by a feedforward
compensation
\begin{align}
\begin{aligned}F_{fr,d} & =r_{v,T}v_{T,d}+r_{C,T}\textrm{sign}(v_{T,d})\\
M_{fr,d} & =r_{v,\varphi}\omega_{\varphi,d}+r_{C,\varphi}\textrm{sign}(\omega_{\varphi,d})
\end{aligned}
\label{eq:frict_comp}
\end{align}
of the friction in order to improve the tracking behaviour. The required
reference trajectories for the velocity of the trolley $v_{T}$ and
the angular velocity $\omega_{\varphi}$ of the rope drum can be computed
from the reference trajectories of the flat output $y=(x_{L},y_{L})$
by the inverse of the state transformation \eqref{eq:coord_transf}
and the parameterizing map \eqref{eq:parameterizing_map_gantry}.
The friction compensation \eqref{eq:frict_comp} is simply added to
the values of the control inputs $F$ and $M$ determined by the tracking
control law \eqref{eq:quasi_static_final_control_law} derived in
Section \ref{subsec:Design_of_quasi-static}. Like in \cite{KolarRamsSchlacher:2017},
we use an asymmetric friction compensation, i.e., the friction coefficients
depend on the sign of the velocities $v_{T,d}$ and $\omega_{\varphi,d}$.
In addition to the compensation of the friction, we also use integral
parts (cf. Section \ref{subsec:Extending-the-tracking}) to avoid
steady-state control deviations. As it is common practice for continuous-time
systems, these integral parts are activated only at the end of the
transitions, i.e., after the reference trajectory has already reached
the final desired position.

\subsection{\label{subsec:Comparison_cont_disc}Comparison continuous-time and
	discrete-time flatness-based control}

In this section, we compare the tracking performance of the discrete-time
quasi-static controller designed in Section \ref{subsec:Design_of_quasi-static}
with the continuous-time quasi-static controller of \cite{KolarRamsSchlacher:2017}.
For this purpose, as a benchmark, we use a transition of the load
from the rest position $y_{d,0}=(0.2\,\textrm{m},0.7\,\textrm{m})$
to the rest position $y_{d,T}=(1.0\,\textrm{m},0.5\,\textrm{m})$
in a transition time of $T=1.7\,\textrm{s}$. Before presenting measurement
results, we briefly discuss the design of the reference trajectories
in the continuous-time as well as in the discrete-time case.

In general, a continuous-time flatness-based control law involves
not only the reference trajectory $y_{d}(t)$ itself but also its
time derivatives. Hence, the trajectory needs to be sufficiently often
differentiable. Since we transfer the load between two rest positions,
it must satisfy the initial- and final conditions
\begin{equation}
y_{d}^{j}(0)=y_{d,0}^{j}\,,\hphantom{aa}y_{d}^{j}(T)=y_{d,T}^{j}\label{eq:intial_final_cond_1}
\end{equation}
and

\begin{equation}
\begin{aligned}\tfrac{d^{i}}{dt^{i}}y_{d}^{j}(t)\bigg\vert_{t=0}=\tfrac{d^{i}}{dt^{i}}y_{d}^{j}(t)\bigg\vert_{t=T} & =0\end{aligned}
\label{eq:intial_final_cond_2}
\end{equation}
for $i=1,\dots,r_{j}-1$ and $j=1,\dots,m$. For the gantry crane,
we have both for the $x$- and the $y$-direction four conditions
at $t=0$ as well as four conditions at $t=T$. Thus, the simplest
approach are polynomial trajectories
\begin{align}
\begin{aligned}x_{L,d}(t) & =\sum_{i=0}^{7}c_{x,i}t^{i}\\
y_{L,d}(t) & =\sum_{i=0}^{7}c_{y,i}t^{i}
\end{aligned}
\label{eq:polynomial_traj}
\end{align}
of order seven with suitably chosen coefficients $c_{x,i}$ and $c_{y,i}$.

In the discrete-time case, the control law \eqref{eq:quasi_static_final_control_law}
depends on the desired sequence $y_{d}(k)$ and its forward-shifts.
In contrast to the continuous-time case, differentiability is not
required. However, since we want to achieve a transition between rest
positions, the reference trajectory needs to satisfy the conditions
\begin{align}
\begin{aligned}y_{d}^{j}(0) & =\ldots=y_{d}^{j}(r_{j}-1)=y_{d,0}^{j}\\
y_{d}^{j}(N) & =\ldots=y_{d}^{j}(N+r_{j}-1)=y_{d,T}^{j}
\end{aligned}
\label{eq:disc_ref_traj_initial_final_cond}
\end{align}
for $j=1,\dots,m$ with $T=NT_{s}$, which are the discrete-time counterpart
to \eqref{eq:intial_final_cond_1} and \eqref{eq:intial_final_cond_2}.
Thus, in principle, for the gantry crane we could again use a polynomial
ansatz of order $(7,7)$. However, in the following we simply evaluate
the continuous-time reference trajectory at the time instants $t=kT_{s}$.
Since \eqref{eq:polynomial_traj} is used to connect rest positions,
the conditions \eqref{eq:disc_ref_traj_initial_final_cond} are obviously
satisfied.

Both the continuous-time and the discrete-time tracking controller
have been implemented on a dSPACE\textsuperscript{\textregistered}
real-time system with different sampling times $T_{s}$. The continous-time
control law of \cite{KolarRamsSchlacher:2017} is simply evaluated
at the time instants $t=kT_{s}$. In the first experiment, we compare
both controllers without integral parts and a sampling time of $T_{s}=10\,\textrm{ms}$.
The corresponding measurements are shown in Fig. \ref{fig:Poly10ms_kont_xy}
and Fig. \ref{fig:Poly10ms_disk_xy}. It can be observed that with
both controllers the reference trajectory is almost perfectly tracked.
Fig. \ref{fig:Poly10ms_err} shows the corresponding tracking errors.
The small deviations and stationary errors are caused mainly by the
non-ideal compensation of the friction.
\begin{figure}
	\centering\includegraphics[width=0.47\columnwidth]{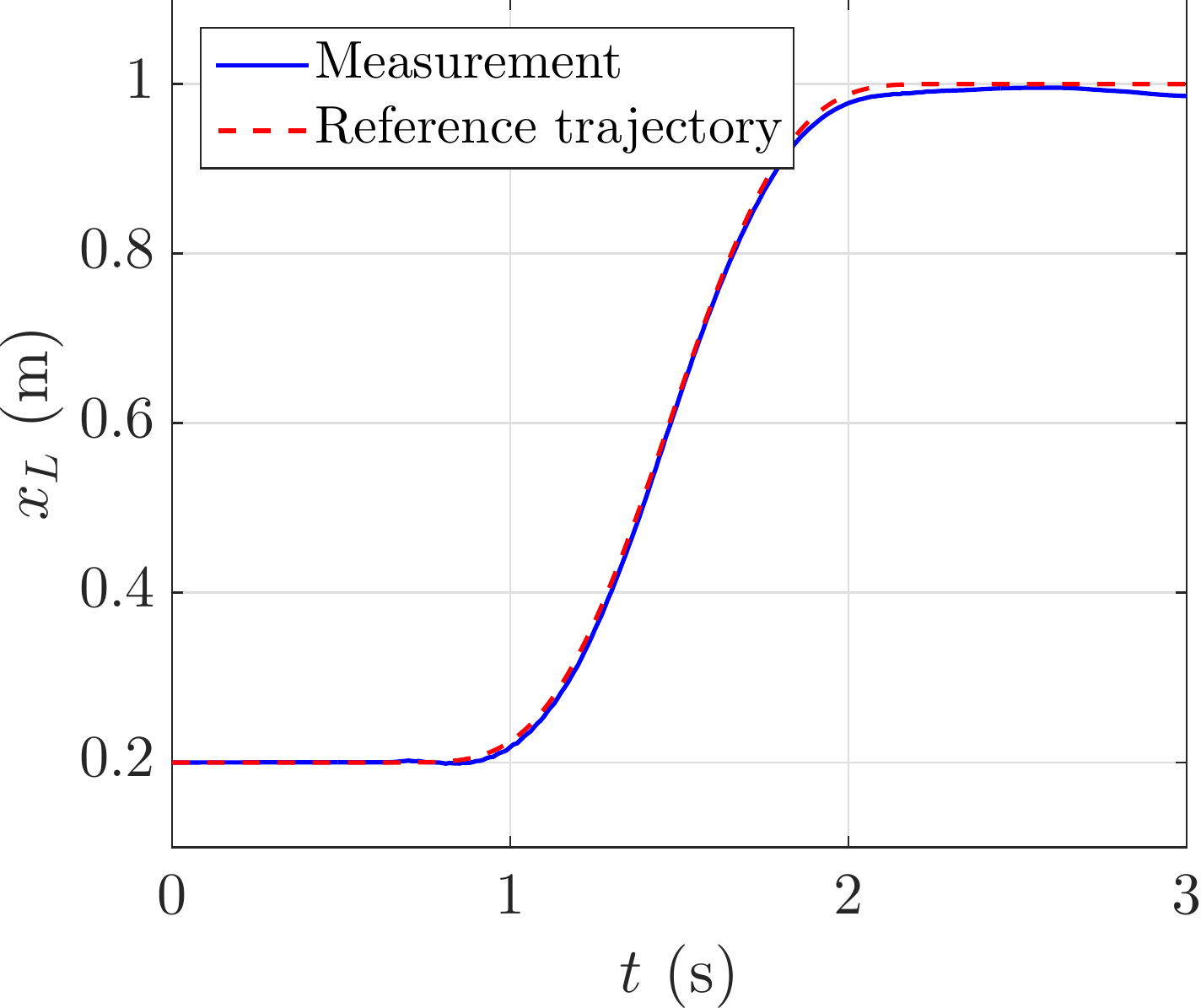}\hspace{0.05\columnwidth}\includegraphics[width=0.47\columnwidth]{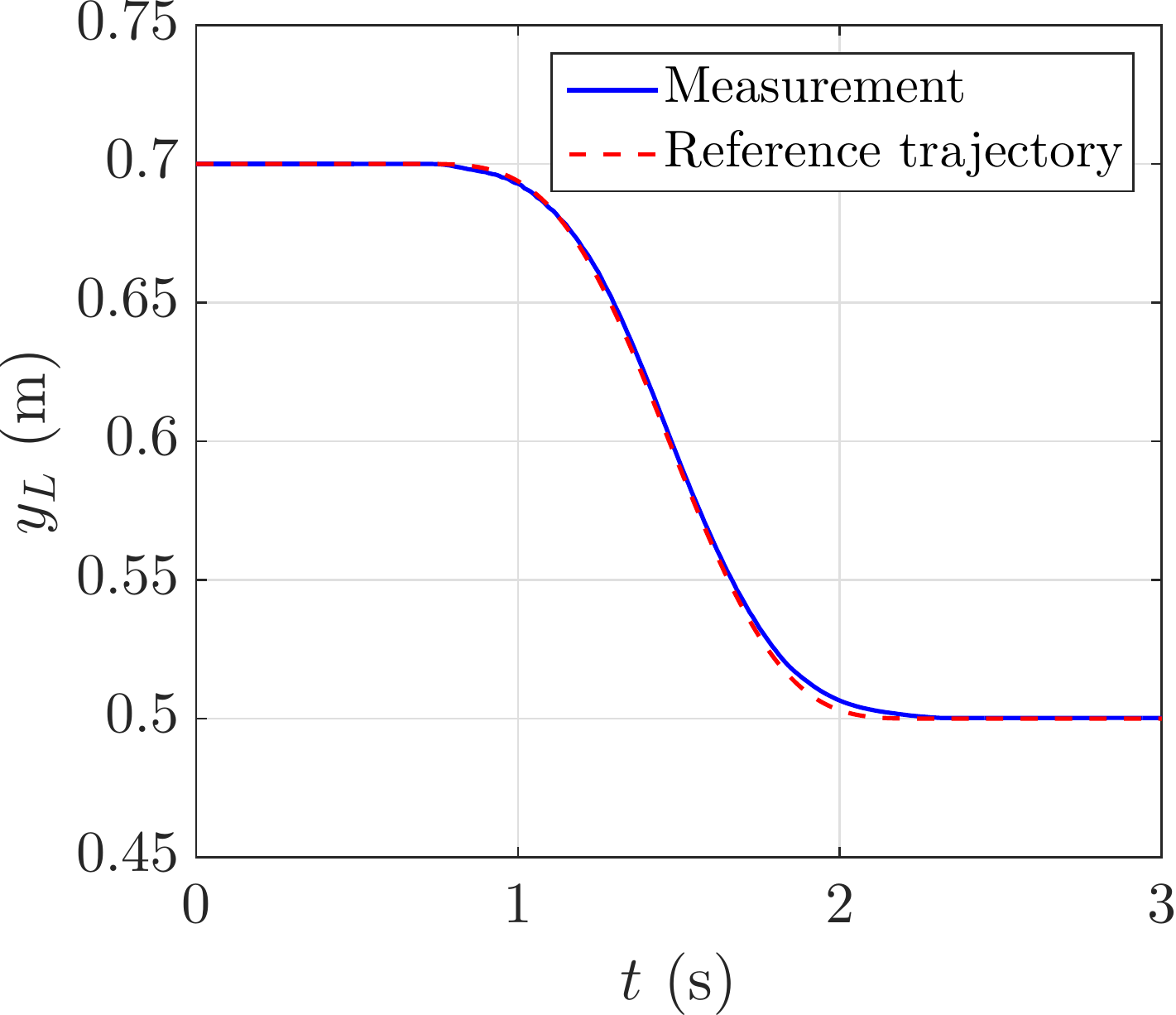}
	
	\caption{\label{fig:Poly10ms_kont_xy}Tracking performance of the continuous-time
		flatness-based controller ($T_{s}=10\,\textrm{ms}$).}
\end{figure}
\begin{figure}
	\centering\includegraphics[width=0.47\columnwidth]{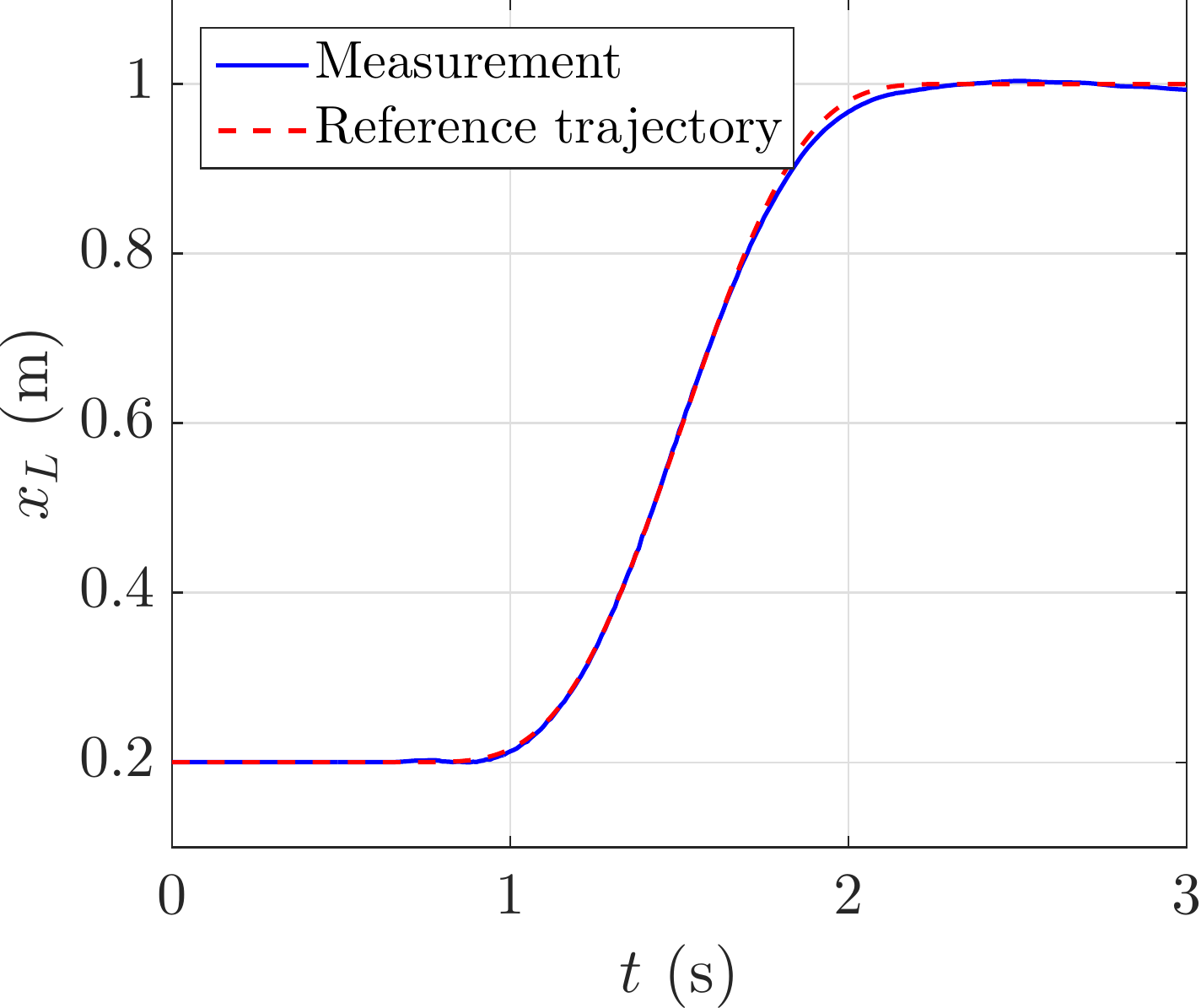}\hspace{0.05\columnwidth}\includegraphics[width=0.47\columnwidth]{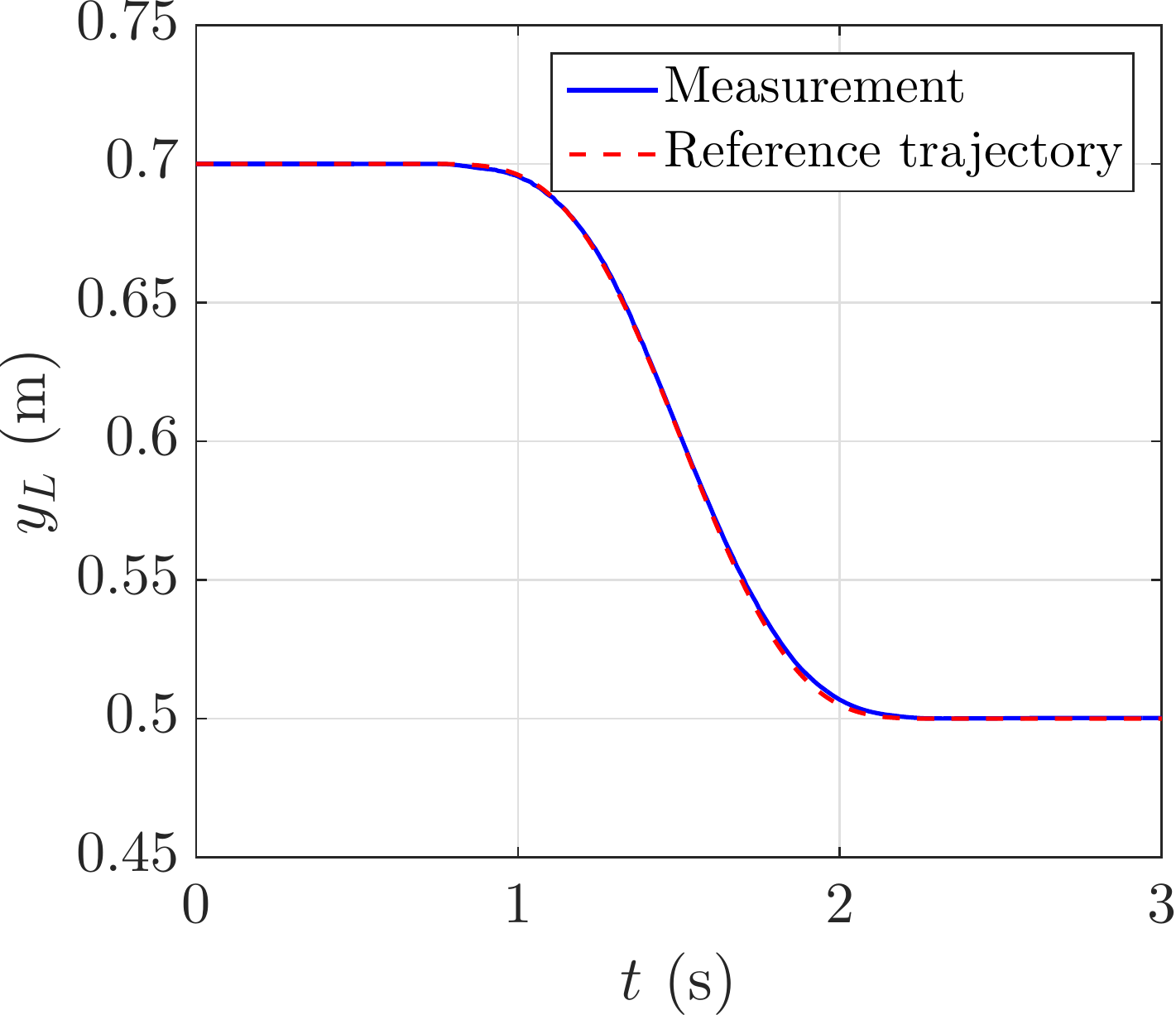}
	
	\caption{\label{fig:Poly10ms_disk_xy}Tracking performance of the discrete-time
		flatness-based controller ($T_{s}=10\,\textrm{ms}$).}
\end{figure}
\begin{figure}
	\centering\includegraphics[width=0.47\columnwidth]{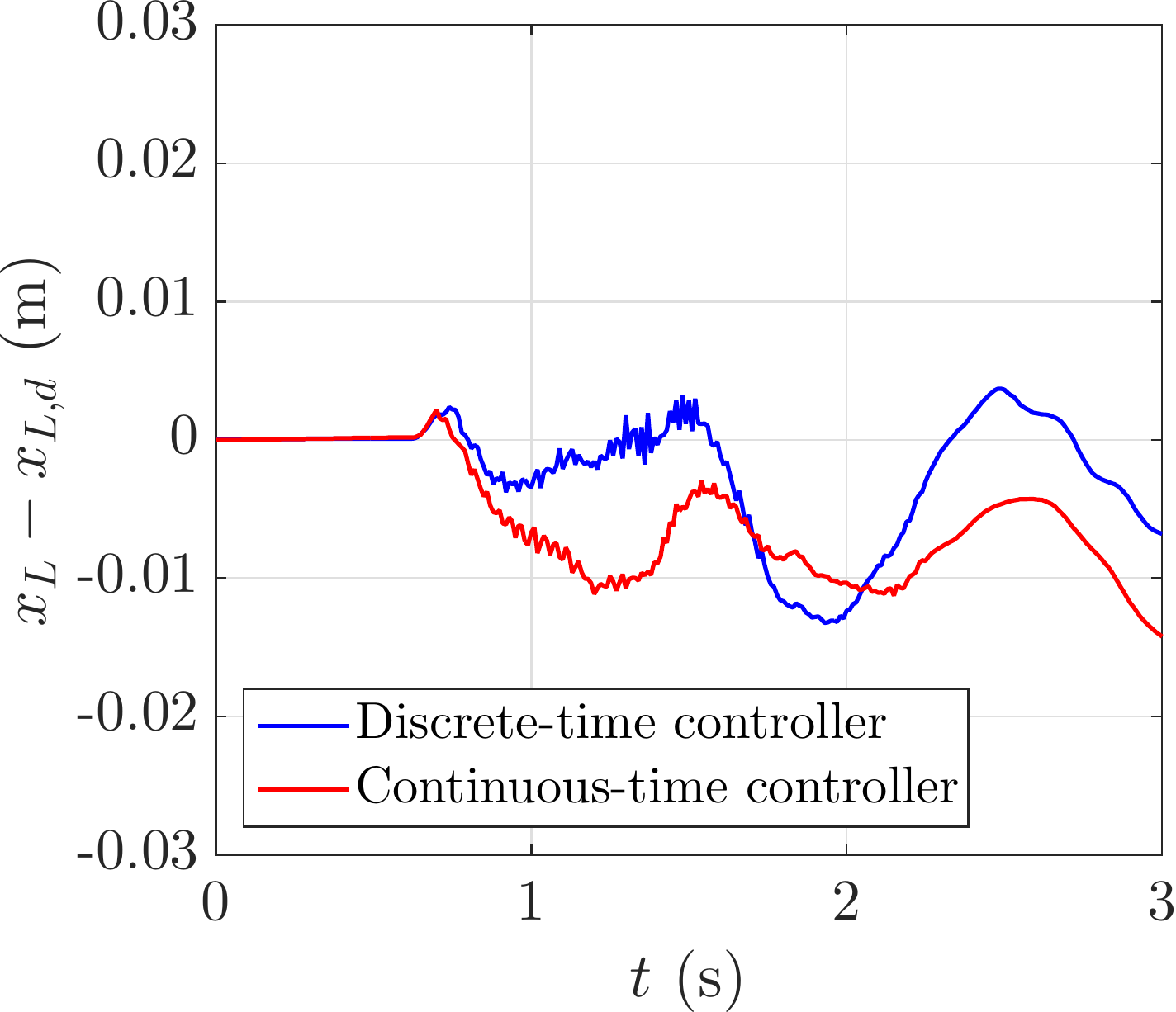}\hspace{0.05\columnwidth}\includegraphics[width=0.47\columnwidth]{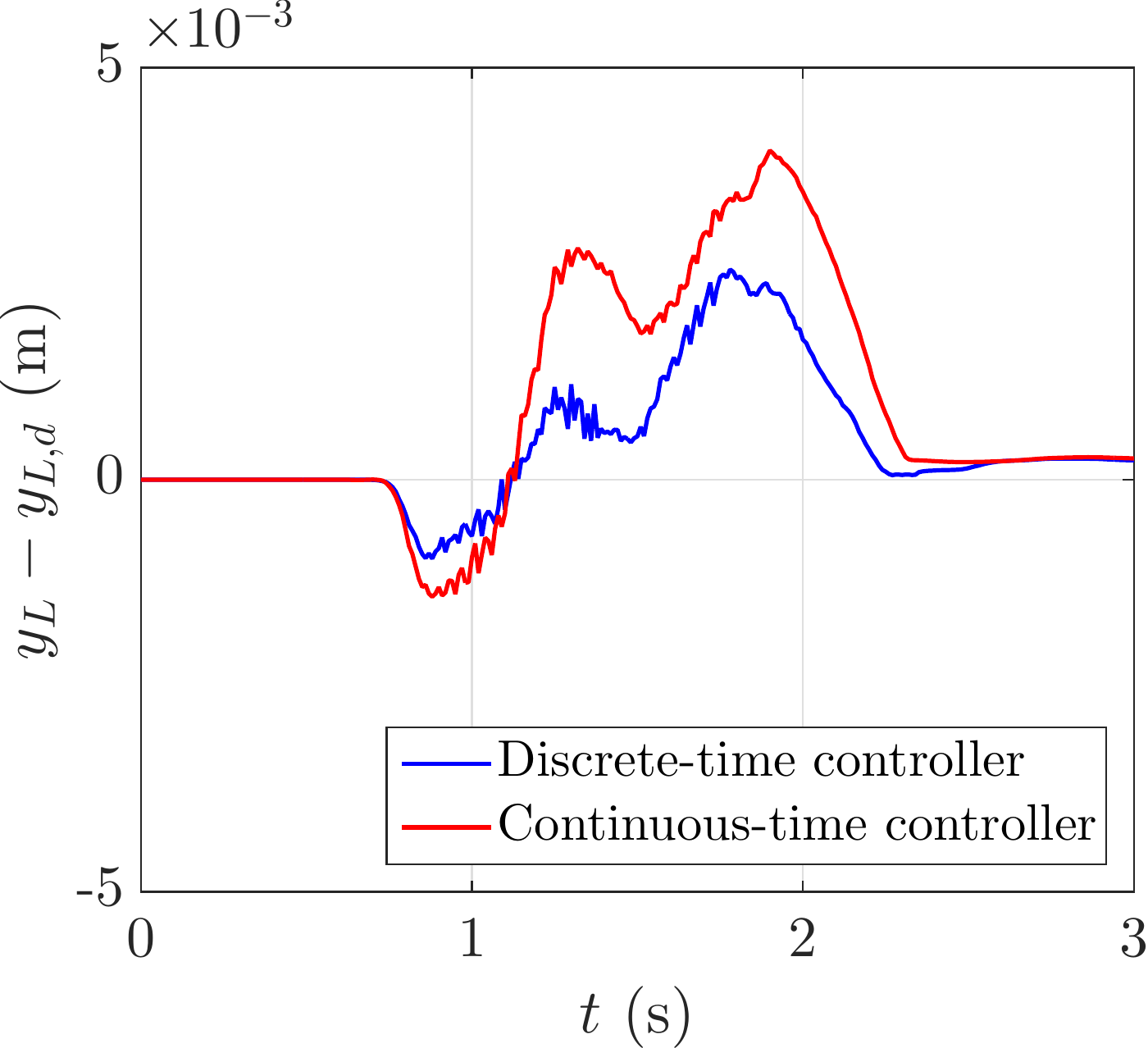}
	
	\caption{\label{fig:Poly10ms_err}Tracking error comparison ($T_{s}=10\,\textrm{ms}$).}
\end{figure}

In the second experiment, the sampling time is increased to $T_{s}=80\,\textrm{ms}$.
It turns out that this sampling time is too large for the continuous-time
controller, whereas the discrete-time controller still transfers the
load between the rest positions as desired. The corresponding measurements
are shown in Fig. \ref{fig:Poly80ms_xy}. However, it can be observed
that the reference trajectory is no longer tracked that well. The
reason is that the Euler-discretization is not exact and the approximation
error increases with the sampling time. Nevertheless, the novel discrete-time
controller seems to be quite robust with respect to high sampling
times (as further experiments up to $T_{s}=100\,\textrm{ms}$ have
shown). In Fig. \ref{fig:Poly80ms_FM}, both the actual control inputs
as well as the pure feedforward part (including the friction compensation)
of the control law are depicted. It can be seen that at the end of
the transition both inputs do not exactly have the expected stationary
values. This is due to the steady-state deviation caused by static
friction and will be discussed in more detail in the next section.
\begin{figure}
	\centering\includegraphics[width=0.47\columnwidth]{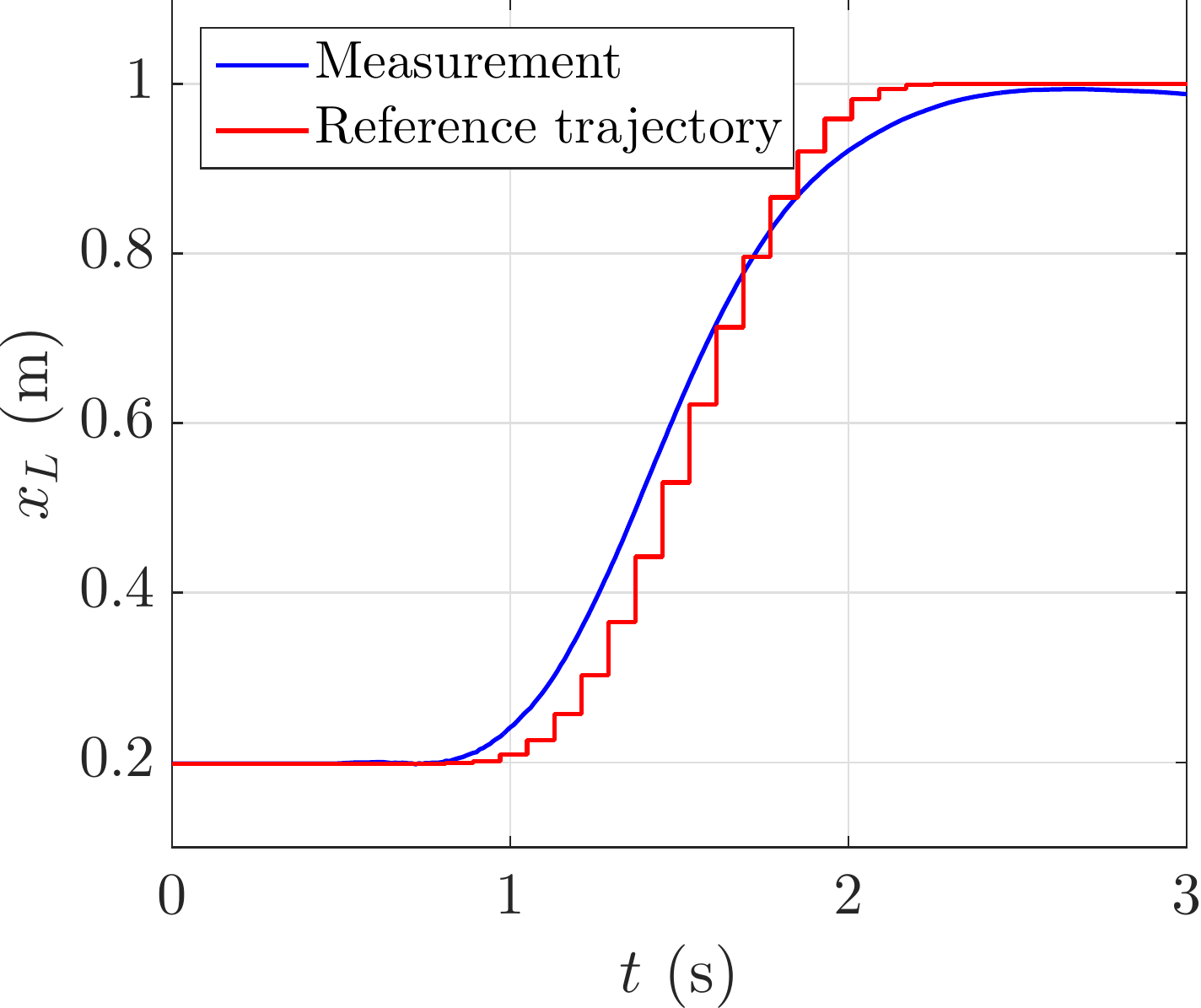}\hspace{0.05\columnwidth}\includegraphics[width=0.47\columnwidth]{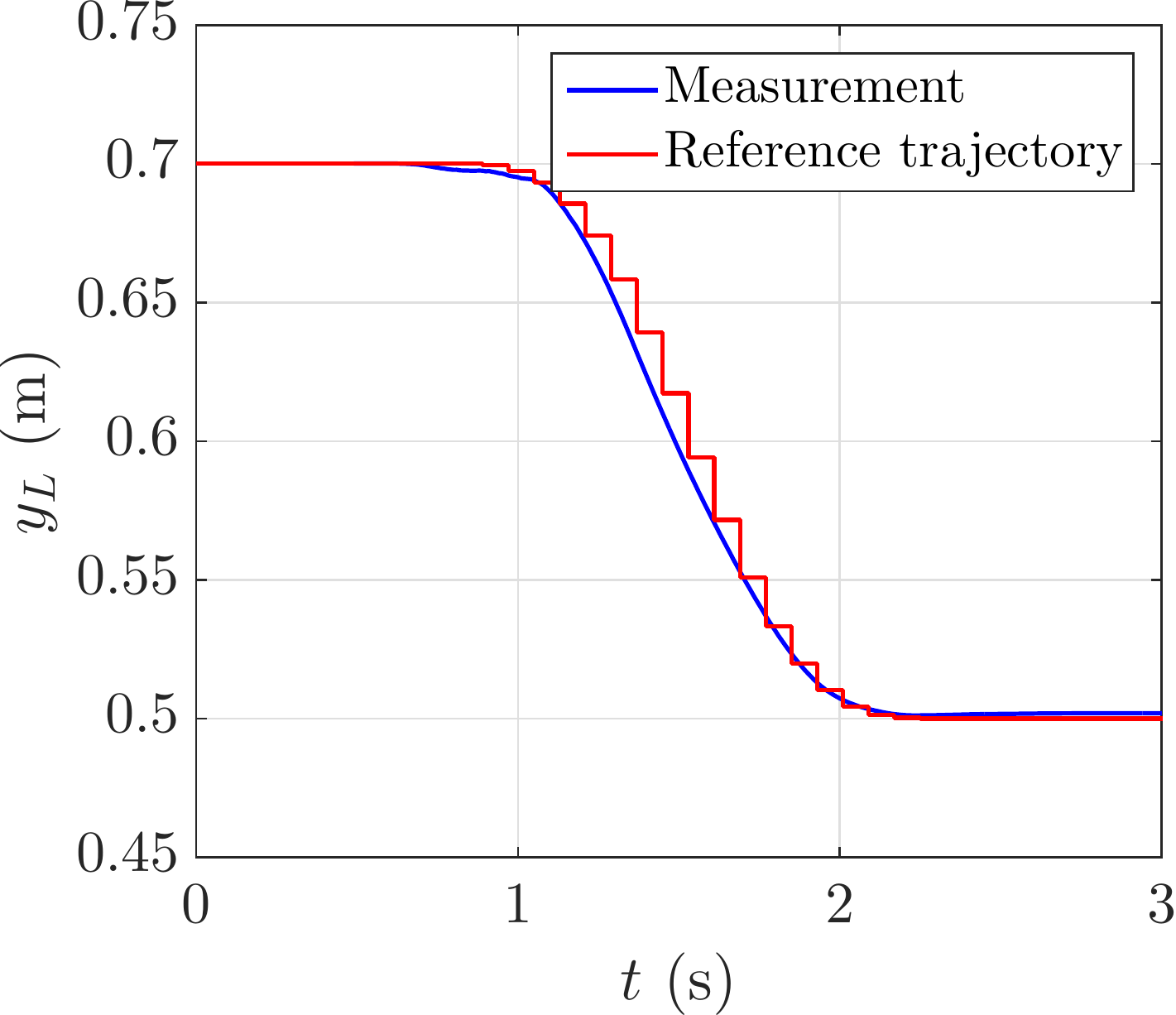}
	
	\caption{\label{fig:Poly80ms_xy}Tracking performance of the discrete-time
		flatness-based controller ($T_{s}=80\,\textrm{ms}$).}
\end{figure}
\begin{figure}
	\centering\includegraphics[width=0.47\columnwidth]{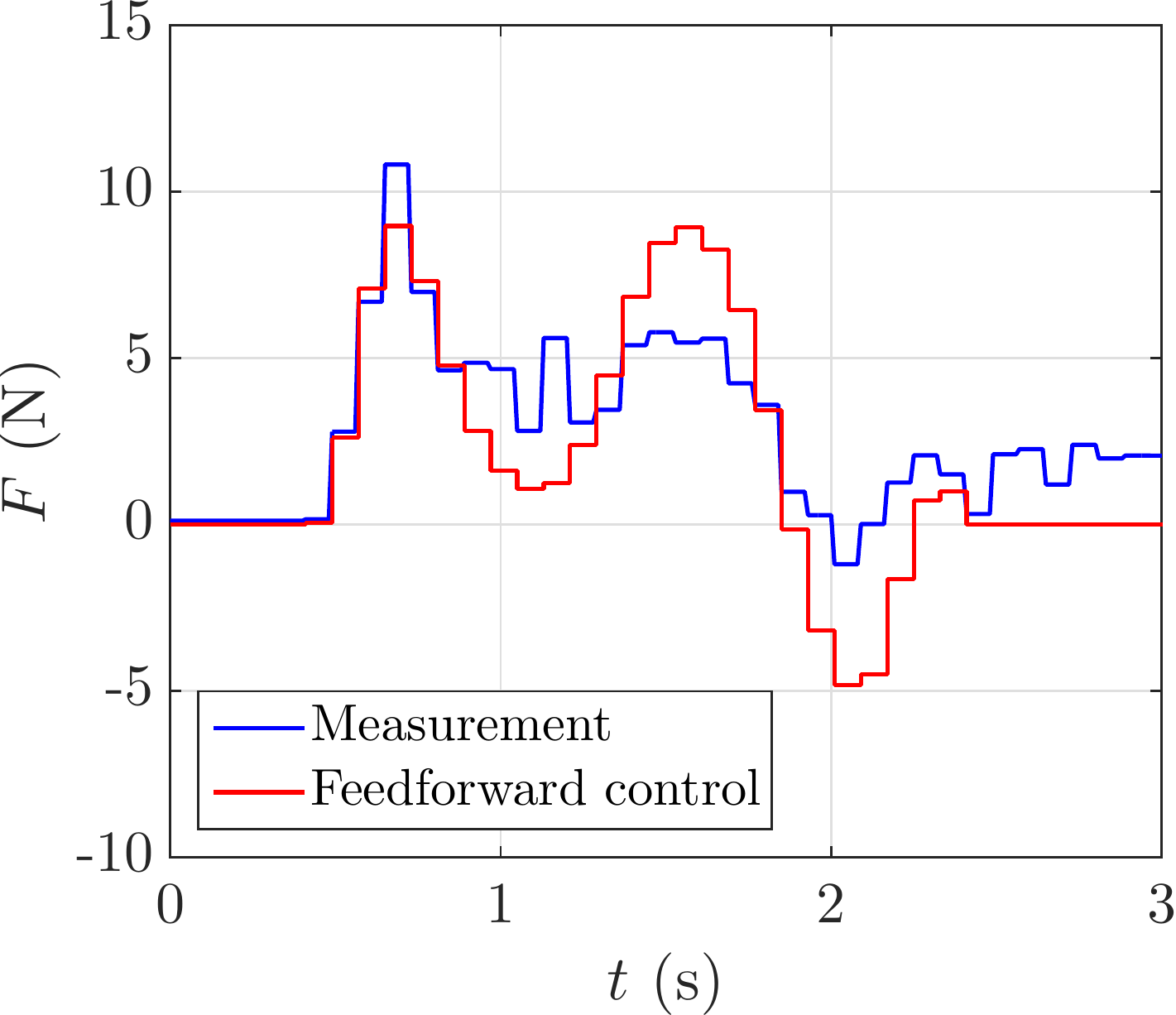}\hspace{0.05\columnwidth}\includegraphics[width=0.47\columnwidth]{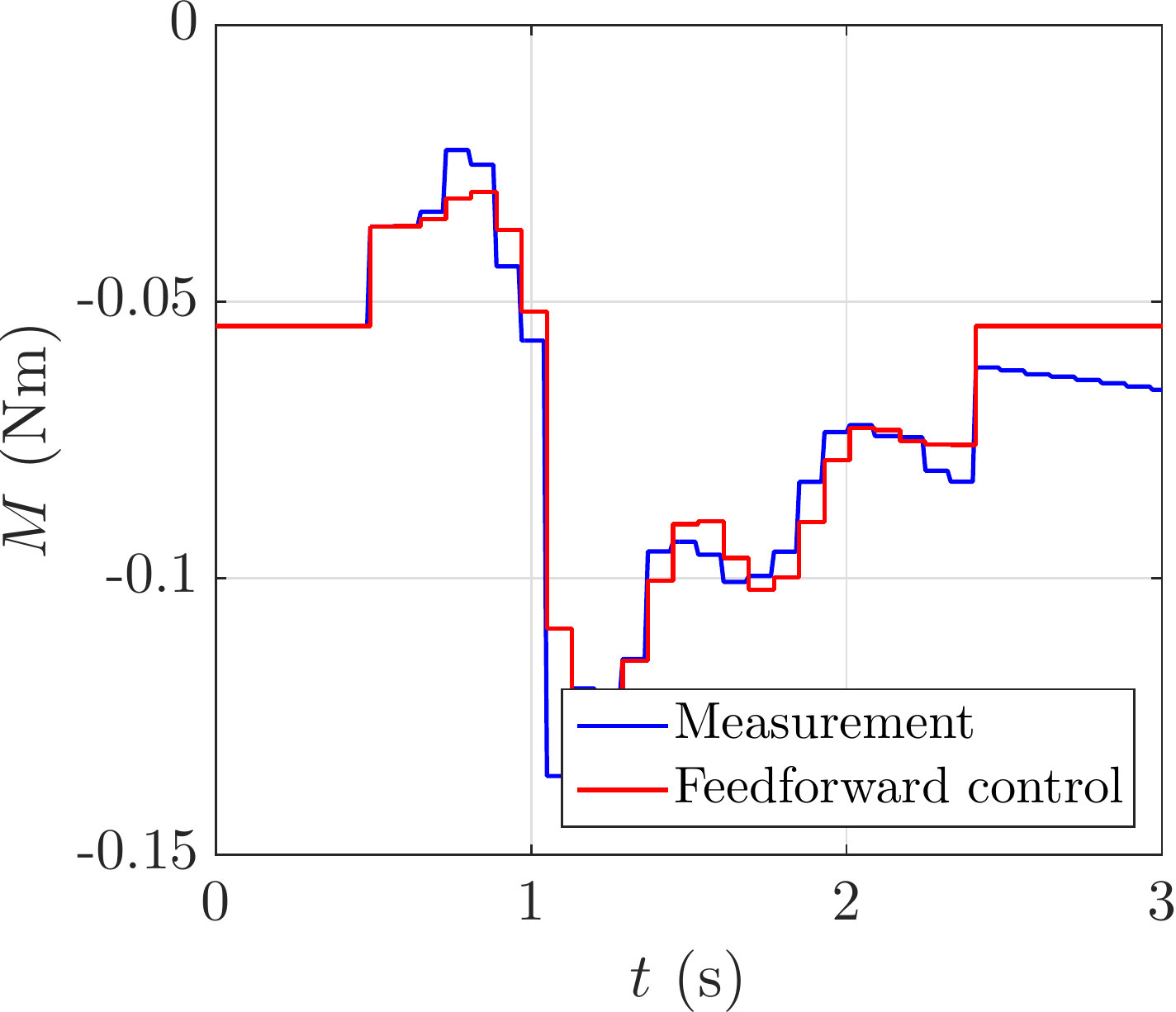}
	
	\caption{\label{fig:Poly80ms_FM}Actual control inputs and feedforward part
		($T_{s}=80\,\textrm{ms}$).}
\end{figure}

\subsection{\label{subsec:Optimization}Design of reference trajectories by optimization}

In the following, we use the flatness of the gantry crane to design
reference trajectories that minimize a given objective function. Instead
of searching for state- and input trajectories, which are constrained
by the system equations, we formulate the optimization problem in
terms of the flat output. As it is well-known, the advantage of the
latter approach is that the trajectories of the flat output are free
and the system equations do not need to be taken into account as constraints.
In the continuous-time case, the possible trajectories of the flat
output belong to an infinite-dimensional space. Thus, in order to
get a finite-dimensional optimization problem, the trajectories are
usually restricted to a set of functions described by a finite number
of parameters, see e.g. \cite{MartinMurrayRouchon:1997}. Possible
approaches are the use of spline curves or the method applied in \cite{KolarRamsSchlacher:2017}.
In the discrete-time case, in contrast, the optimization problem is
inherently finite-dimensional since we are dealing with sequences
instead of functions.

In the following, we formulate and solve such an optimization problem
for the sampled-data model \eqref{eq:sampled_data_model_gantry_crane}
of the gantry crane. The goal is to achieve a transition between the
two rest positions of Section \ref{subsec:Comparison_cont_disc} in
$T=3\,\textrm{s}$ while minimizing the absolute acceleration
\[
a_{L}=\sqrt{f_{a_{L,x}}^{2}+f_{a_{L,y}}^{2}}
\]
of the load. The functions $f_{a_{L,x}}$ and $f_{a_{L,y}}$ are the
same as in \eqref{eq:ganry_new_state}. By the use of the parameterizing
map \eqref{eq:parameterizing_map_gantry}, $a_{L}$ can now be expressed
as a function of the flat output and its forward-shifts. The optimization
problem can be formulated as
\begin{align*}
\begin{aligned}\underset{p}{\textrm{\ensuremath{\mathrm{min}}}} & \,\text{\ensuremath{\varepsilon}}\\
\text{\ensuremath{\mathrm{s.t.}}}\, & a_{L}(k)\leq\ensuremath{\varepsilon}\\
& |F(k)|\leq F_{\textrm{max}}\\
& |M(k)|\leq M_{\textrm{max}}
\end{aligned}
\end{align*}
with $k=0,\dots,N-1$. The inputs have been limited by $F_{\textrm{max}}=10\,\text{N}$
and $M_{\textrm{max}}=0.2\,\textrm{Nm}$.\footnote{On the laboratory setup of the gantry crane the actual input limitations
	are higher. We use the above limits for the optimization to ensure
	that there is enough margin for the additional friction compensation
	and the feedback part of the control law.} Since the load shall be transferred from one rest position to another,
the reference trajectory needs to satisfy the conditions \eqref{eq:disc_ref_traj_initial_final_cond}
and the vector of optimization variables is reduced to $p=[\ensuremath{\varepsilon},y_{d}(4),\dots,y_{d}(N-1)]$.
In Fig. \ref{fig:Opt_XY_10ms}, the calculated optimal trajectory
for a sampling time of $T_{s}=10\,\textrm{ms}$ as well as the polynomial
trajectory, which was used as initial value for the optimization,
are shown. The absolute value $a_{L}$ of the acceleration of the
load is shown in Fig. \ref{fig:Opt_Acc_10ms} for both the polynomial
as well as the optimized trajectory. It can be observed that for the
optimal trajectory the acceleration is constant during almost the
whole transition. The corresponding input trajectory $F(k)$ temporarily
reaches the limits $\pm F_{\text{max}}$ at the beginning, in the
middle, and at the end of the transition, while the constraint $|M(k)|\leq M_{\textrm{max}}$
is never active.

Figures \ref{fig:Opt_xy_10ms} to \ref{fig:Op_FM_10ms} show measurement
results for the determined optimal trajectory. In Fig. \ref{fig:Opt_xy_10ms},
it can be observed that the reference trajectories are again almost
perfectly tracked. Particularly interesting are the horizontal and
vertical velocities of the load depicted in Fig. \ref{fig:Opt_vxvy_10ms},
which increase and decrease for most of the time with a constant slope.
This is in accordance with the almost constant acceleration of Fig.
\ref{fig:Opt_Acc_10ms}. The corresponding control inputs are shown
in Fig. \ref{fig:Op_FM_10ms}. It can be seen that because of the
additional friction compensation already the feedforward part slightly
exceeds the limit $|F(k)|\leq F_{\textrm{max}}$. However, the resulting
control input, including the feedback part, remains within the input
limitations of the laboratory setup. At the end of the transition
at $t=4\,\textrm{s}$, it can again be observed that the inputs $F$
and $M$ do not exactly have the expected stationary values because
the end position of the load slightly differs from its desired position.
Thus, the controller attempts to compensate this error. Since the
resulting force $F$ and torque $M$ are too small to overcome the
static friction of the trolley and the rope drum, these input values
would be applied permanently. For this experiment, however, we have
used the extended control law with integral parts, which are activated
after the transition of the reference trajectory is completed. Consequently,
as can be observed in Fig. \ref{fig:Op_FM_10ms}, the applied inputs
increase after $t=4\,\textrm{s}$ until they are large enough to overcome
the static friction.
\begin{figure}
	\centering\includegraphics[width=0.47\columnwidth]{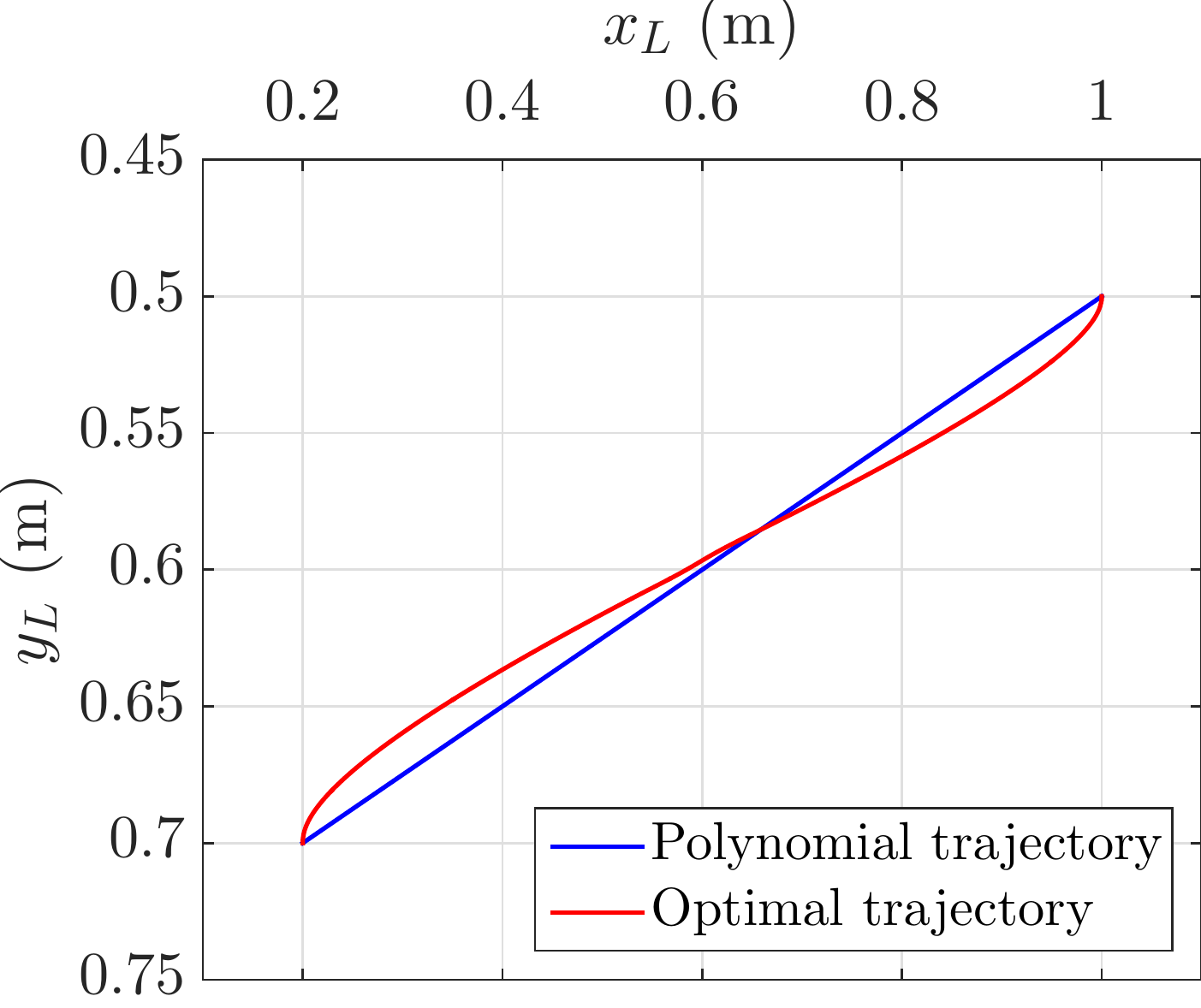}
	
	\caption{\label{fig:Opt_XY_10ms}Optimal trajectory ($T_{s}=10\,\textrm{ms}$).}
\end{figure}
\begin{figure}
	\centering\includegraphics[width=0.47\columnwidth]{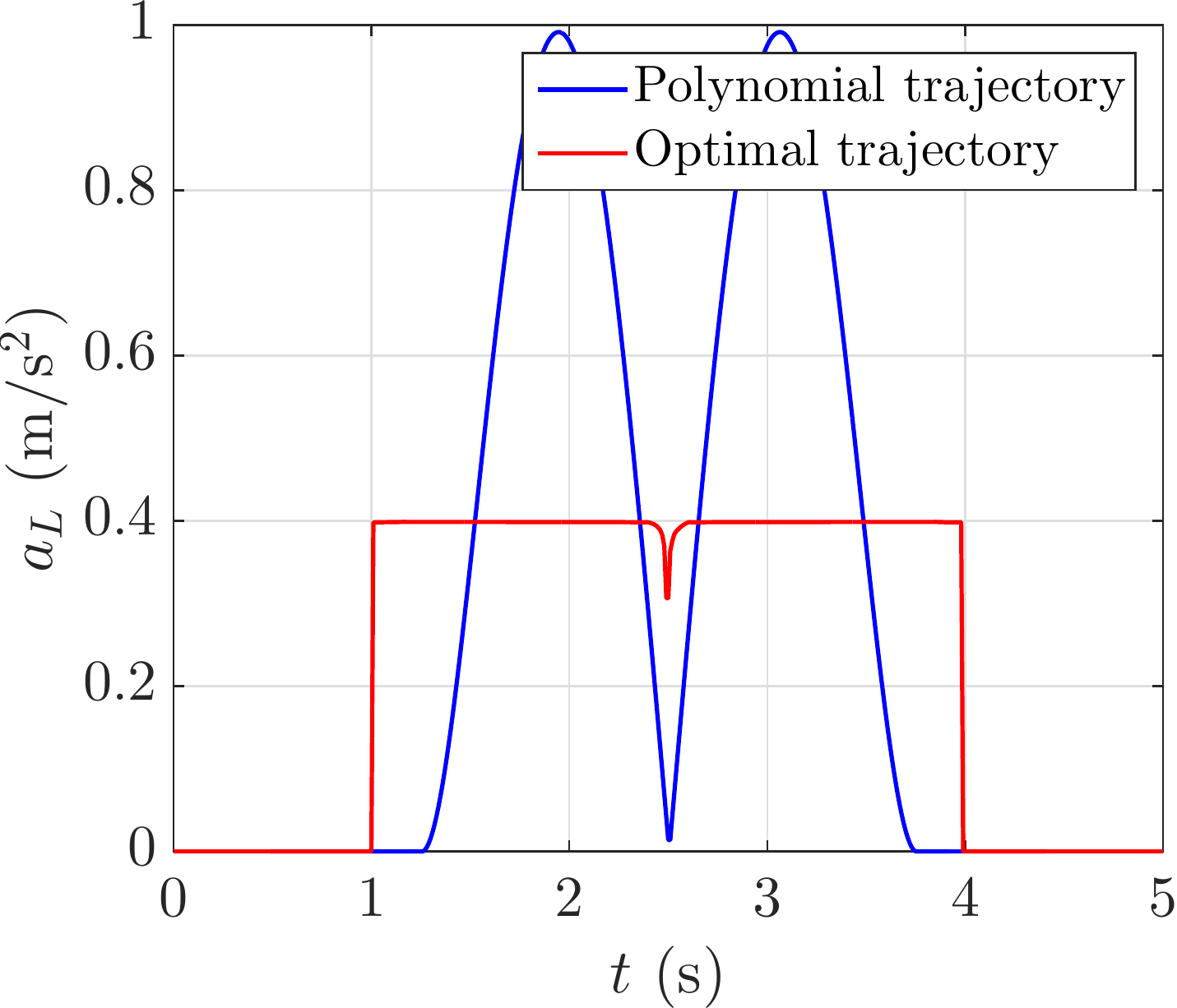}
	
	\caption{\label{fig:Opt_Acc_10ms}Acceleration for the optimal trajectory ($T_{s}=10\,\textrm{ms}$).}
\end{figure}
\begin{figure}
	\centering\includegraphics[width=0.47\columnwidth]{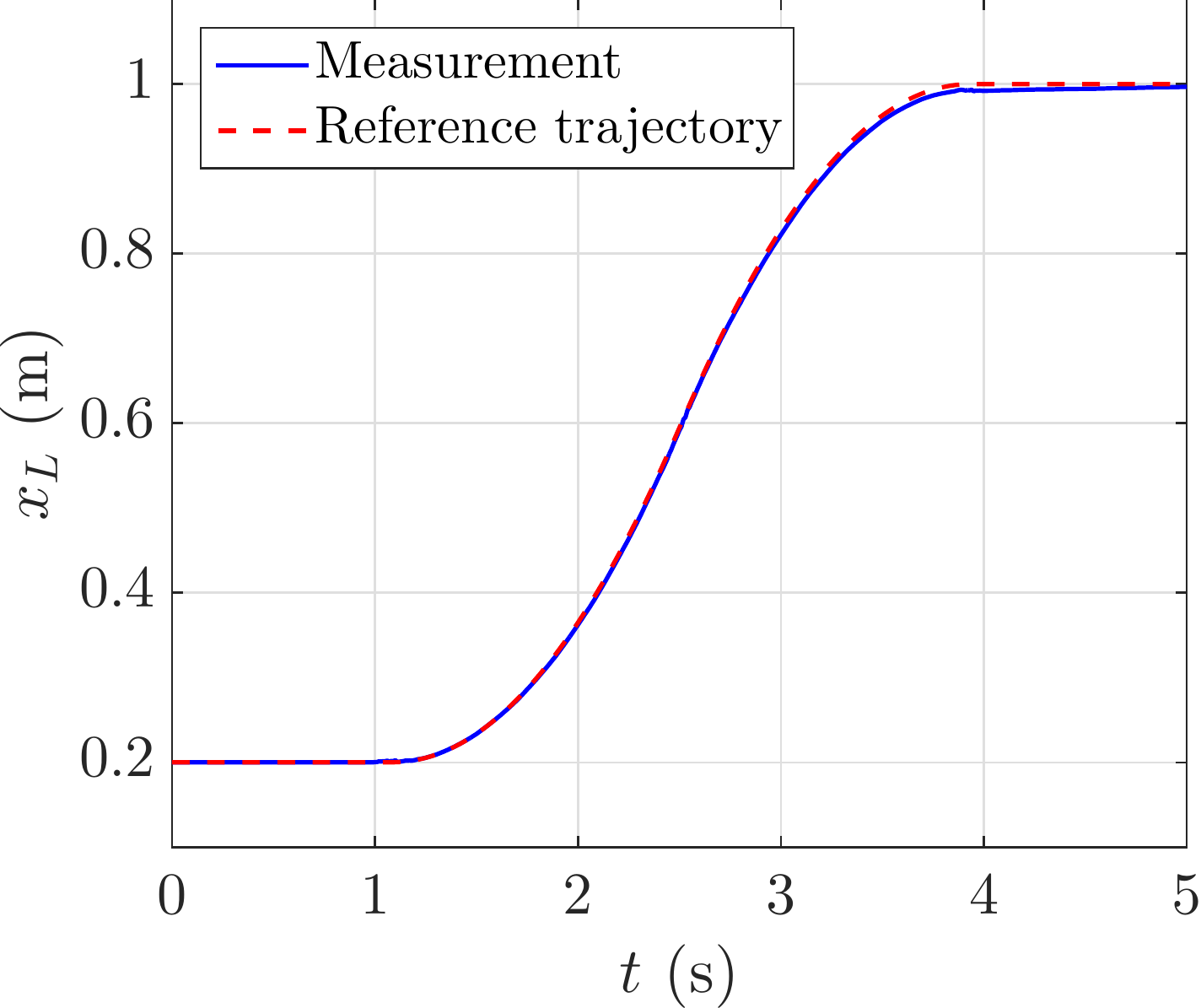}\hspace{0.05\columnwidth}\includegraphics[width=0.47\columnwidth]{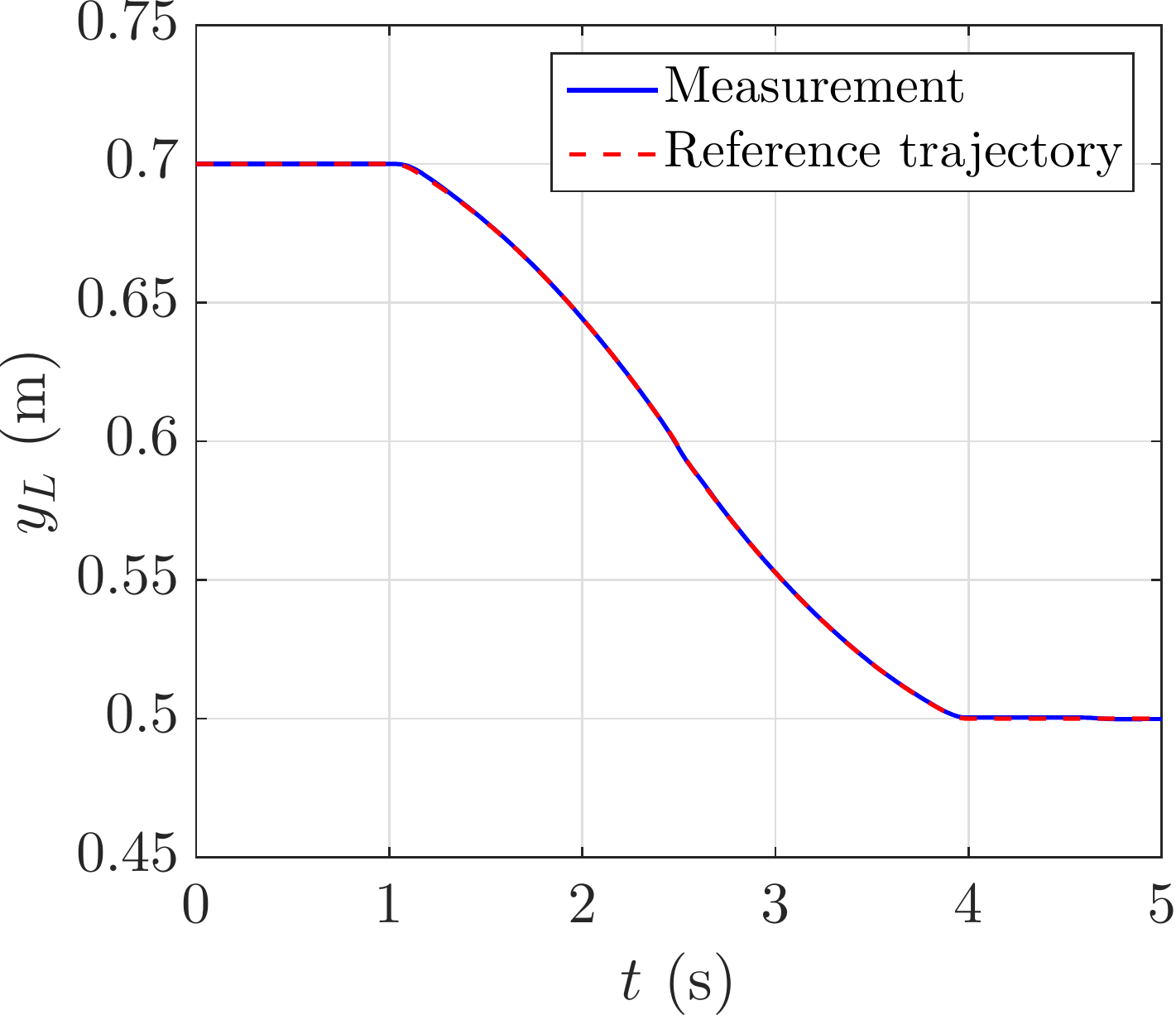}
	
	\caption{\label{fig:Opt_xy_10ms}Position tracking performance for the optimal
		trajectory ($T_{s}=10\,\textrm{ms}$).}
\end{figure}
\begin{figure}
	\centering\includegraphics[width=0.47\columnwidth]{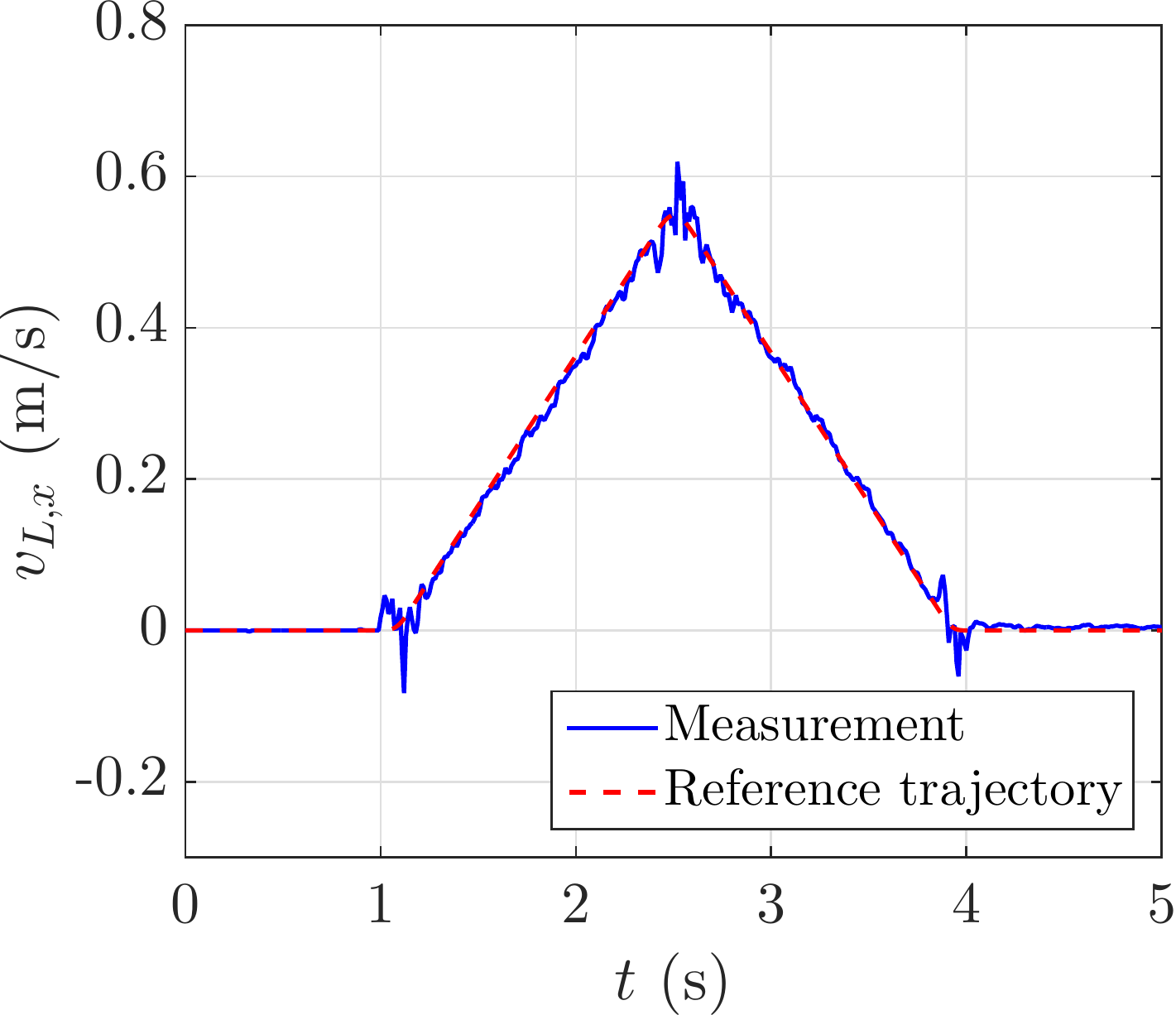}\hspace{0.05\columnwidth}\includegraphics[width=0.47\columnwidth]{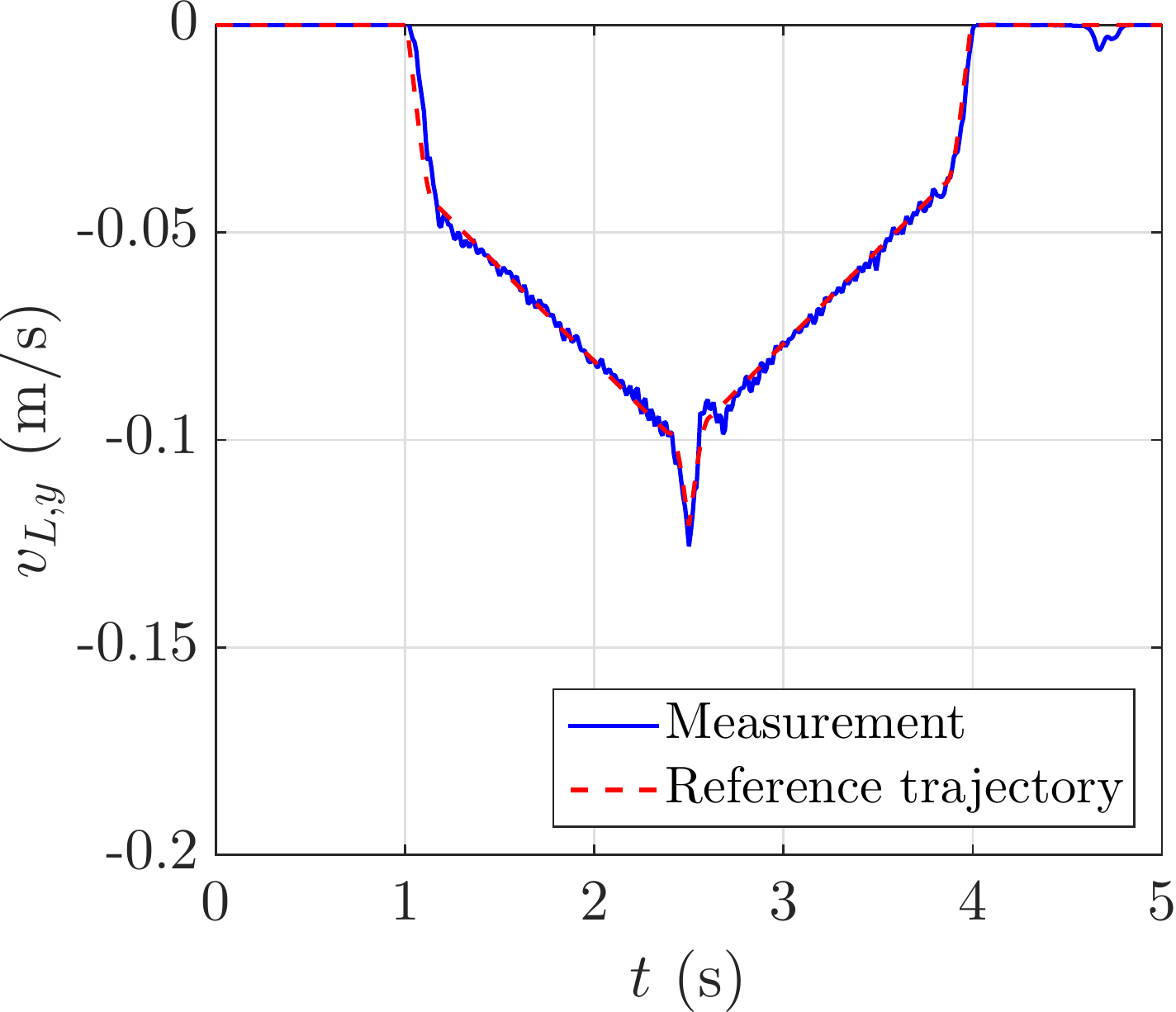}
	
	\caption{\label{fig:Opt_vxvy_10ms}Velocity tracking performance for the optimal
		trajectory ($T_{s}=10\,\textrm{ms}$).}
\end{figure}
\begin{figure}
	\centering\includegraphics[width=0.47\columnwidth]{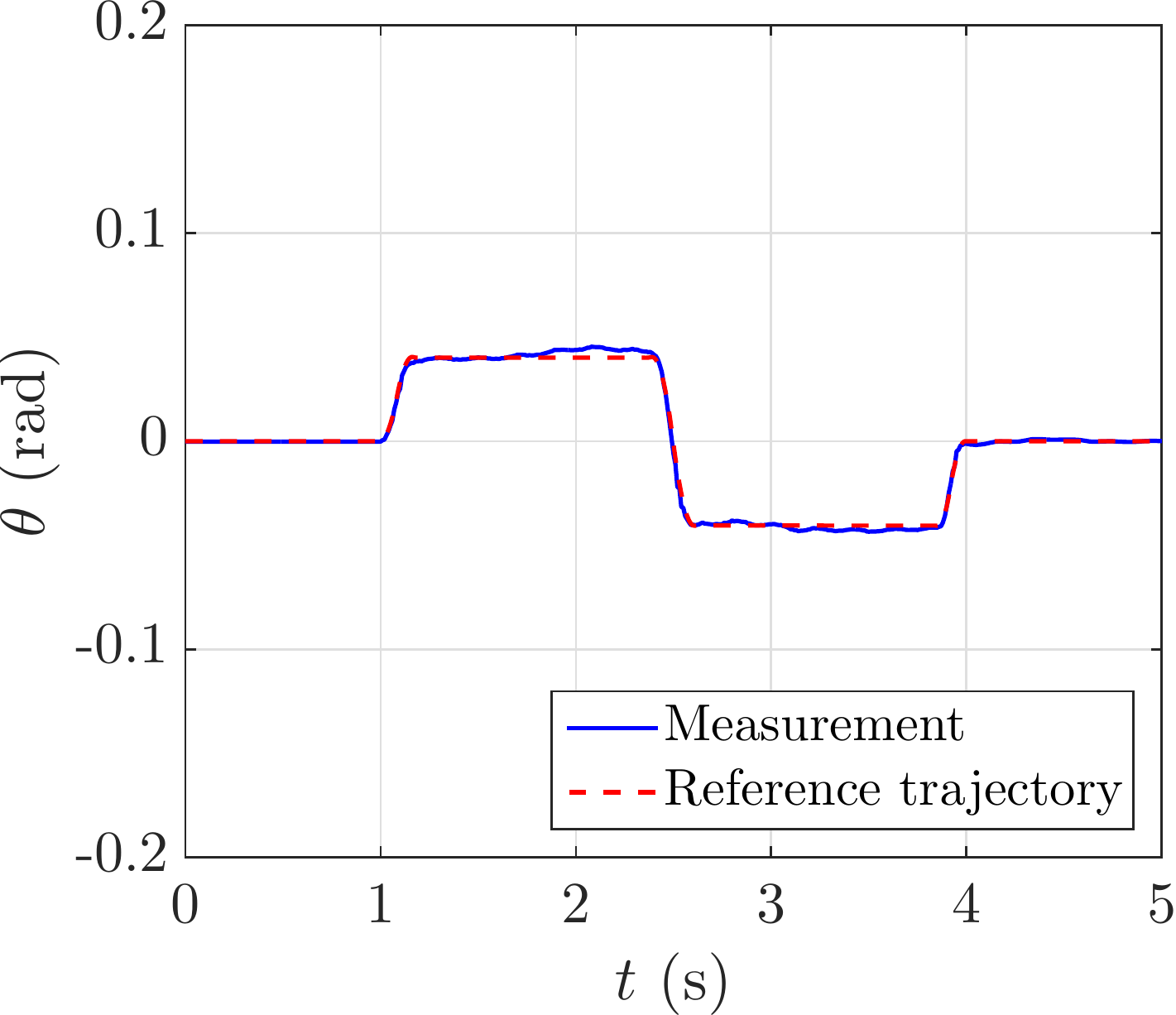}\hspace{0.05\columnwidth}\includegraphics[width=0.47\columnwidth]{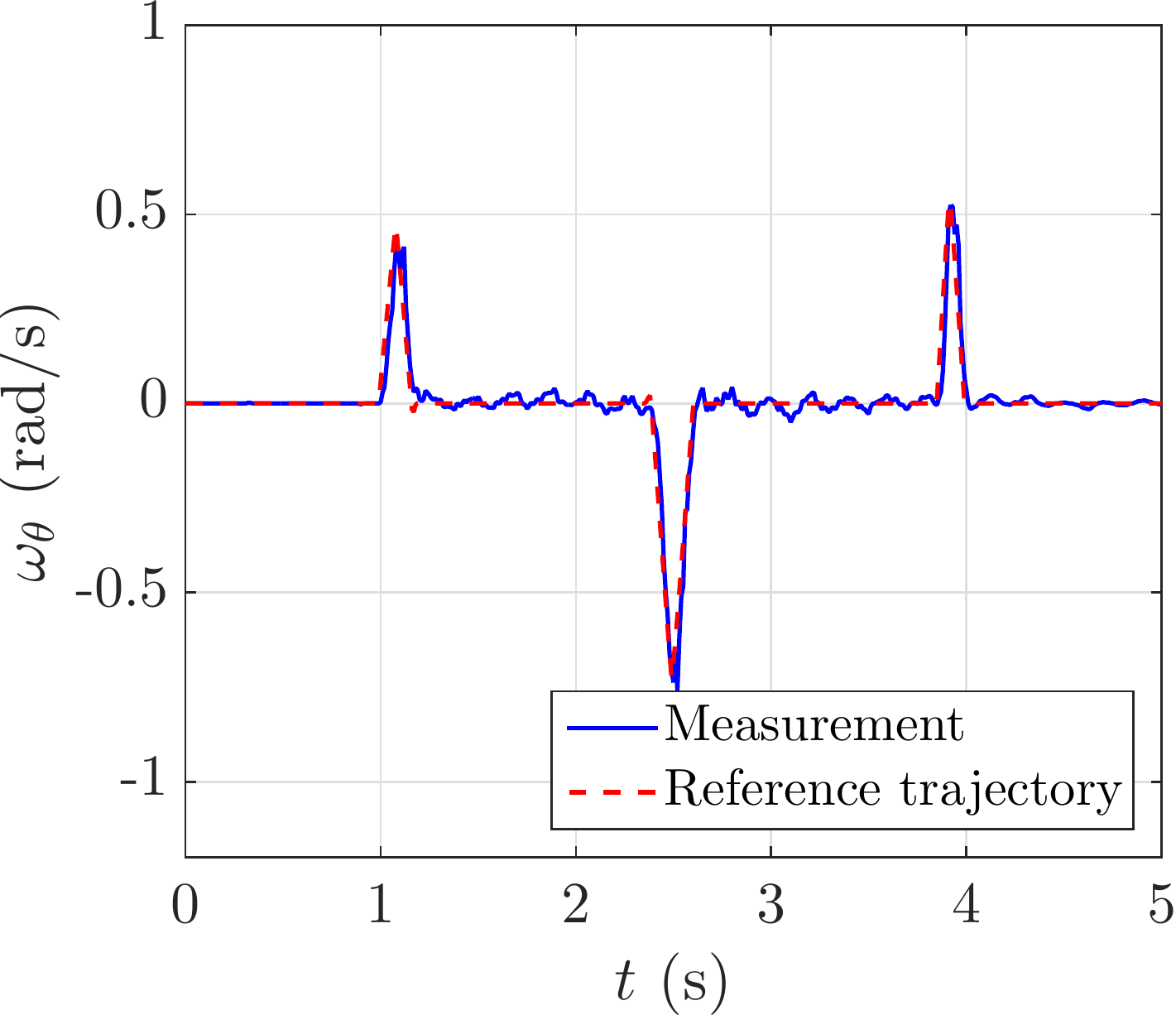}
	
	\caption{\label{fig:Opt_theta_omega_theta_10ms}Angle and angular velocity
		tracking performance for the optimal trajectory ($T_{s}=10\,\textrm{ms}$).}
\end{figure}
\begin{figure}
	\centering\includegraphics[width=0.47\columnwidth]{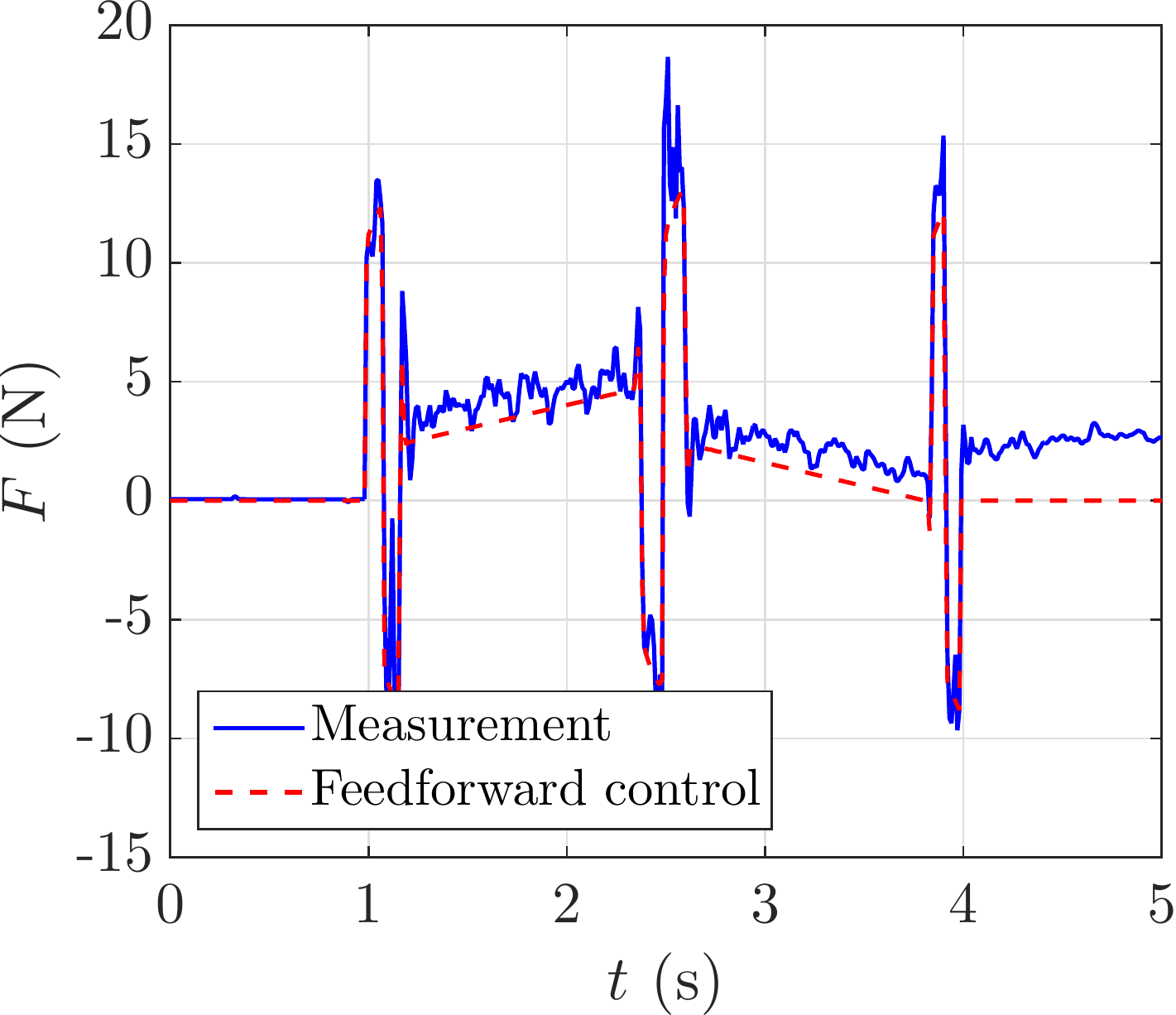}\hspace{0.05\columnwidth}\includegraphics[width=0.47\columnwidth]{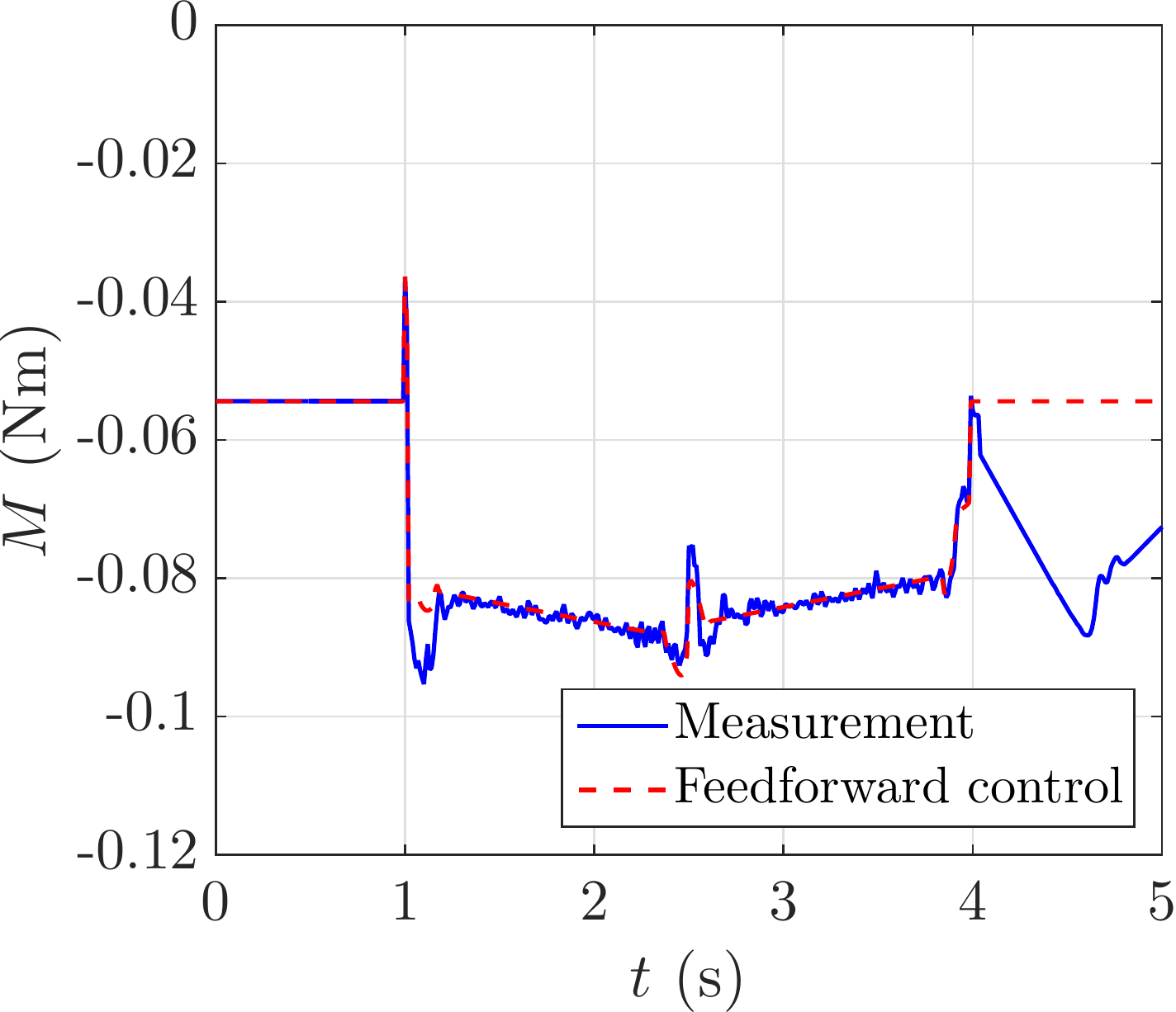}
	
	\caption{\label{fig:Op_FM_10ms}Actual control inputs and feedforward part
		for the optimal trajectory ($T_{s}=10\,\textrm{ms}$).}
\end{figure}

\subsection{\label{subsec:Lying_eight}Tracking a closed path}

As already mentioned in Section \ref{subsec:Comparison_cont_disc},
in the continuous-time case the reference trajectory $y_{d}(t)$ needs
to be sufficiently often differentiable. Moreover, in order to track
a reference trajectory $y_{d}(t)$, all time derivatives up to $y_{d}^{(r)}(t)$
must be computed beforehand. In the discrete-time case, in contrast,
there are no restrictions on the reference trajectory $y_{d}(k)$
such as differentiability, and instead of time derivatives the control
law involves only forward-shifts of $y_{d}(k)$. This simplifies the
implementation, since these forward-shifts can be constructed easily
by chains of unit delays.

To show the flexibility of the discrete-time approach, in the last
experiment a closed path with the shape of a lying eight is used as
reference trajectory. In Fig. \ref{fig:Eight_10ms_xy}, it can be
seen that $y_{d}(k)$ is tracked well and that the stationary errors
at the end of the reference trajectory are compensated by the integral
parts. Fig. \ref{fig:Eight_10ms_XY} shows the desired and the measured
trajectory in the $xy$-plane. The small deviations at the reversal
points of the load are caused by the strong friction of the trolley
on the rail. 
\begin{figure}
	\centering\includegraphics[width=0.47\columnwidth]{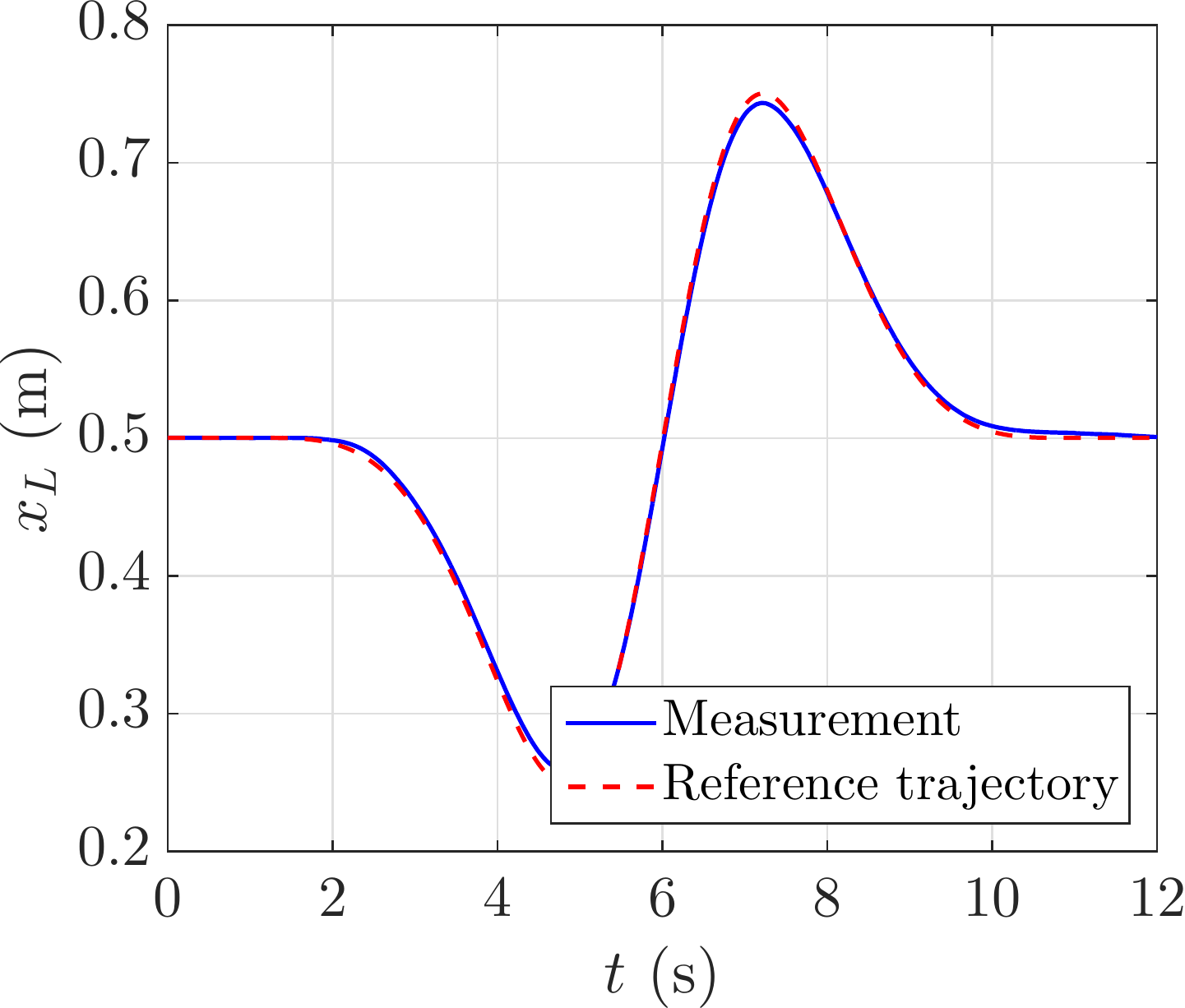}\hspace{0.05\columnwidth}\includegraphics[width=0.47\columnwidth]{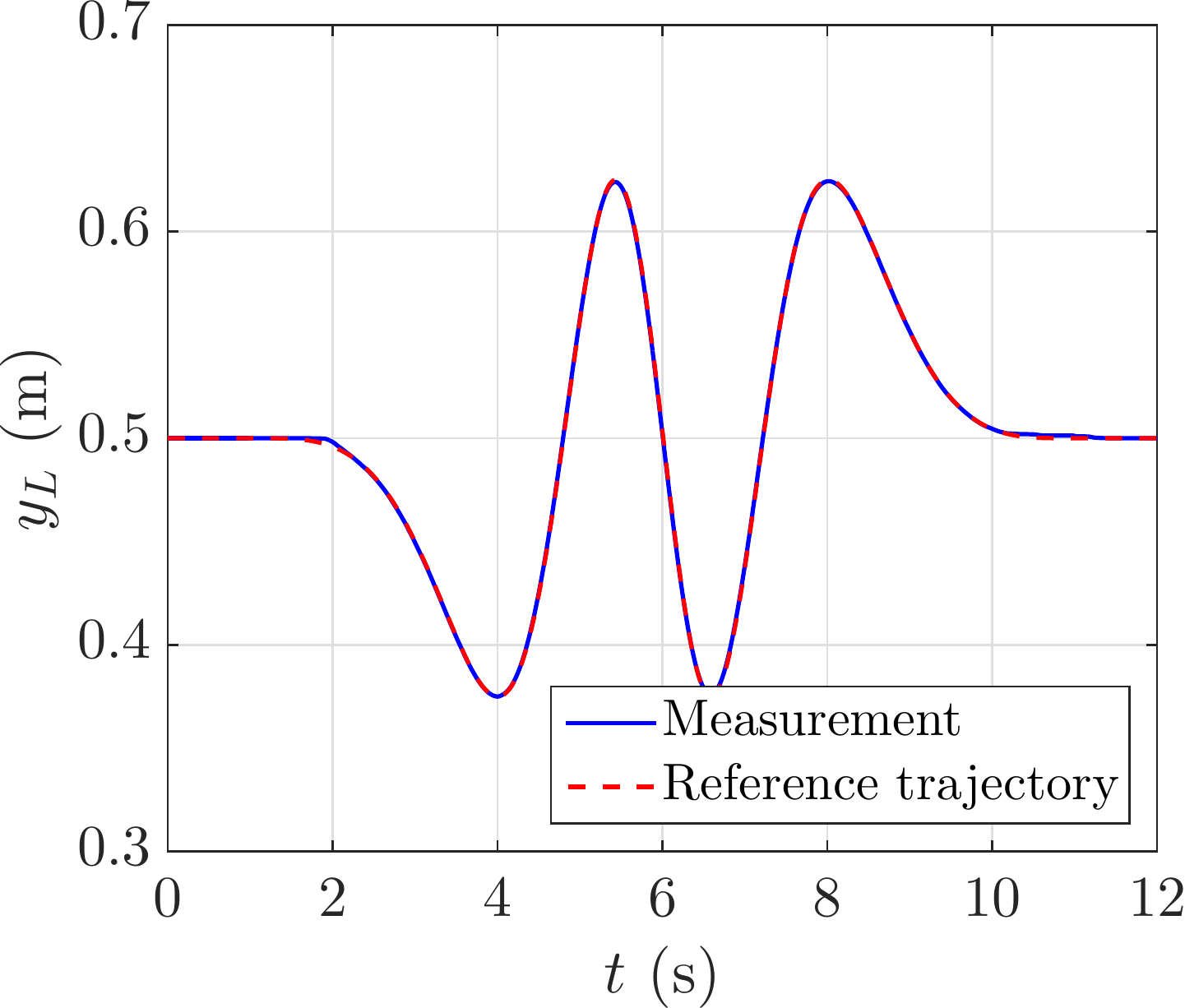}
	
	\caption{\label{fig:Eight_10ms_xy}Position tracking performance for a closed
		path ($T_{s}=10\,\textrm{ms}$).}
\end{figure}
\begin{figure}
	\centering\includegraphics[width=0.47\columnwidth]{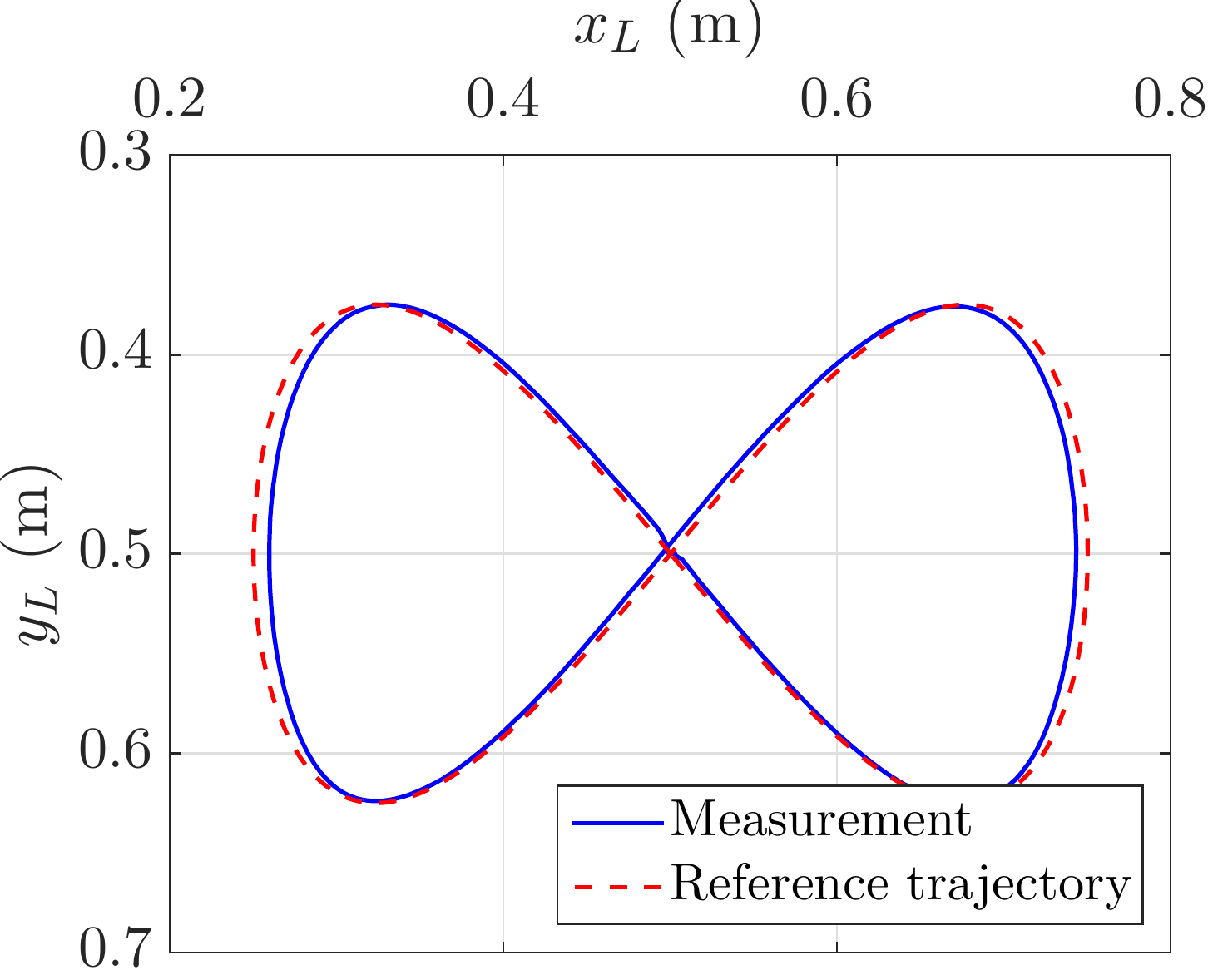}
	
	\caption{\label{fig:Eight_10ms_XY}Tracking a closed path ($T_{s}=10\,\textrm{ms}$).}
\end{figure}

\section{Conclusion}

We have shown by means of the laboratory model of a gantry crane that
the concept of discrete-time flatness is well-suited for the design
of tracking controllers in practical applications. For nonlinear systems,
the difficulty lies in determining a flat sampled-data model of the
plant. In order to overcome this problem, we have proposed a constructive
method which relies on a combination of a suitable state transformation
and a subsequent Euler-discretization. Although the Euler-discretization
is not exact, experiments with the gantry crane have shown that the
novel discrete-time flatness-based controller can be used over a significantly
wider range of sampling rates than its continuous-time counterpart.
Moreover, we have shown that the property of discrete-time flatness
is also beneficial in the context of trajectory planning by optimization.

An important topic for future research are alternative, more precise
discretization methods which also preserve the flatness. The goal
is to achieve a better tracking performance for even lower sampling
rates. In contrast to the proposed approach via the Euler-discretization,
such methods may possibly lead to discrete-time systems which are
flat in the more general sense of \cite{DiwoldKolarSchoeberl2022}
and not necessarily forward-flat.

\bibliography{Bibliography_Johannes_April_2020}

\end{document}